\documentclass[11pt]{amsart}
\usepackage{amssymb,amsmath,latexsym,amscd,amsfonts,amsrefs,color}
\usepackage{graphics}
\usepackage{epsfig}
\usepackage{psfrag}
\usepackage{braket}
\usepackage{enumerate}
\usepackage{accents}
\def\ignore #1 {}

\usepackage{color}
\usepackage[all]{xy}
\usepackage{braket}
\usepackage{slashed}
\usepackage[colorlinks=true, linkcolor=red,citecolor=blue]{hyperref}

 \newtheorem{theorem}{Theorem}
\newtheorem{lemma}[theorem]{Lemma}
\newtheorem{proposition}[theorem]{Proposition}
\newtheorem{corollary}[theorem]{Corollary}

\theoremstyle{definition}
\newtheorem{definition}[theorem]{Definition}

\theoremstyle{remark}
\newtheorem{remark}[theorem]{Remark}

\numberwithin{equation}{section}
\numberwithin{theorem}{section}
 
 \def\CC{\mathbb{C}}
  \def\RR{\mathbb{R}}

\def\V{\mathcal{V}}

\def\A{\mathcal{A}}

\def\L{\mathcal{L}}

\def\N{\mathcal{N}}

\def\R{\mathbb{R}}

\def\U{\mathcal{U}}
\def\V{\mathcal{V}}
\def\ZZ{\mathbb{Z}}

\DeclareMathOperator\End{End}

\DeclareMathOperator\Hom{Hom}

\DeclareMathOperator\ch{ch}

\DeclareMathOperator\even{even}
\DeclareMathOperator\odd{odd}

\DeclareMathOperator\Ind{Ind}

\DeclareMathOperator{\Tr}{Tr}

\DeclareMathOperator{\Ima}{Im}
\DeclareMathOperator\Sym{Sym}
\DeclareMathOperator\Spin{Spin}

\DeclareMathOperator\Fr{Fr}

\title[An Analytical approach to the equivariant index and Witten genus]{An analytical approach to the equivariant index and Witten genus on spin manifolds}
\author{Juan Jose Villarreal}
\address{Virginia Commonwealth University,
Richmond, VA, USA}
\email{juanjos3villarreal@gmail.com}

\begin{document}

\maketitle

\begin{abstract}
This work is divided in two cases. In the first case, we consider a spin manifold $M$ as the set of fixed points of an $S^{1}$-action on a spin manifold $X$, and in the second case we consider the spin manifold $M$ as the set of fixed points of an $S^{1}$-action on the loop space of $M$. For each case, we build on $M$ a vector bundle, a connection and a set of bundle endomorphisms. These objects are used to build global operators on $M$ which define an analytical index in each case. In the first case, the analytical index is equal to the topological equivariant Atiyah-Singer index, and in the second case the analytical index is equal to a topological expression where the Witten genus appears.  
\end{abstract}


\section{Introduction}

Let $X$ be an oriented, compact, even dimensional, spin Riemannian manifold. On $X$ there is a bundle of complex spinors $\bigtriangleup(X)$, and a Dirac operator $\slashed{D}_{X}$ acting on sections of this bundle. The orientation on $X$ defines a Riemannian volume form, and its image in the Clifford algebra splits the bundle of complex spinors $\bigtriangleup(X)=\bigtriangleup^{+}(X)\oplus \bigtriangleup^{-}(X)$, we call $\bigtriangleup^{+}(X)$ the positive spinor bundle and $ \bigtriangleup^{-}(X)$ the negative spinor bundle. The index of the Dirac operator is calculated using the Atiyah-Singer Index theorem 
$$\Ind (\slashed{D}_{X})=\dim \ker (\slashed{D}_{X}|_{\bigtriangleup^{+}(X)})-\dim \ker (\slashed{D}_{X}|_{\bigtriangleup{-}(X)})=\int_{M} A(X),$$ 
where $A(X)$ is the topological $\hat{A}$ class, the notation $|_{E}$ means that the operator is restricted to the sections $\Gamma(E)$ for a vector bundle $E$. For an introduction to the Atiyah-Singer index theorems see \cite{lawson2016spin, atiyah1968index}. Additionally, If $R\rightarrow X$ is a complex bundle equipped with a connection, then we have a twisted Dirac operator $\slashed{D}_{X}\otimes R:\Gamma (\bigtriangleup^{+}(X)\otimes R)\rightarrow \Gamma(\bigtriangleup^{-}(X)\otimes R)$ and the Atiyah-Singer index theorem states
\begin{equation}\Ind (\slashed{D}_{X}\otimes R)=\dim \ker (\slashed{D}_{X}\otimes R|_{\bigtriangleup^{+}(X)\otimes R})-\dim \ker (\slashed{D}_{X}\otimes R|_{\bigtriangleup^{-}(X)\otimes R})=\int_{X} A(X)\ch (R),\label{TAS}\end{equation}
where $\ch(R)$ is the Chern character of the bundle $R$.   

Now, let us assume that there is an $S^{1}$-action orientation  preserving isometry acting on $X$ such that the action lifts to the spin bundle. Then this action preserves the splitting $\bigtriangleup(X)=\bigtriangleup^{+}(X)\oplus \bigtriangleup^{-}(X)$ and commutes with the Dirac operator. The equivariant index is defined by
\begin{equation}\Ind (\slashed{D}_{X},q)=\Tr (q|_{\ker (\slashed{D}_{X}|_{\bigtriangleup^{+}(X)})})-\Tr (q|_{\ker (\slashed{D}_{X}|_{\bigtriangleup^{-}(X)})}),\label{equiin}\end{equation}
where $q=e^{i\theta}\in\mathbb{C}^{\times}$ and $q|_{E}$ denotes the $S^{1}$-action on sections $\Gamma(E)$ of the vector bundle $E$. 

There is also a topological formula for the analytical equivariant index above, in order to express this formula we consider the following definitions and properties, see \cite{atiyah1970spin} for details. The set of fixed points of the $S^{1}$-action on $X$ is a manifold that we denote by $M$. We assume that $M$ is connected. The normal bundle $N'$ of $M$ in $X$ can be written as a direct sum $N'=\bigoplus _{r\in \mathcal{A}}N'_{r}$ where each $N'_{r}$ is a real bundle equipped with a complex structure and $\mathcal{A}$ is a finite set of positive integers, see Theorem \ref{Koba2}. We denote by $N'_{r}$ the underlying real bundle and by $N_{r}$ the complex bundle. On $N_{r}$ the $S^{1}$-action acts as $q^{r}$. Note that the fixed point manifold $M$ has an induced orientation from the complex bundles $N_{r}$ and the orientation on $X$. We define $S_{t}(E):=1+tE+t^{2}\Sym^{2}(E)+\cdots $ as the formal power series with values in vector bundles on $M$.  Then the equivariant index in \eqref{equiin} is given by  
\begin{equation}\Ind (\slashed{D}_{X},q)= q^{\frac{1}{2}c_{1}}\int_{M} A(M)\ch (\sqrt{\det N}\otimes_{r\in\mathcal{A}} S_{q^{r}}(N_{r})), \label{ASS}\end{equation}
where $c_{1}:=\sum_{r\in\mathcal{A}}rd_{r} $ for $d_{r}:=\dim_{\mathbb{C}}N_{r}$. Note that $q$ is not branched because $\sum_{r\in\mathcal{A}}r d_{r}=0$ (mod $2$), this follows from the lifting of the action to the spin bundle on $X$. The topological expression above makes sense even when the bundle $\det N$ does not have a square root bundle. In this paper, we consider the particular case where $M$ is a spin manifold, in this case the square root of $\det N$ does exists (Proposition \ref{pro2.2}). Finally, note that the Atiyah-Hirzebruch theorem states that \eqref{ASS} is indeed zero. Even with this result the expression above can be easily generalized for more interesting cases, the twisted cases. Our interest in $\eqref{ASS}$ is that this expression can be seen as a finite dimensional analogous of the Witten genus.


For an introductory review of elliptic and Witten genus see \cite{hirzebruch1992manifolds}. Let $M$ be a compact, oriented manifold such that $\dim M=4k$, where $k$ is a positive integer, the Witten genus is defined as follows  
 $$\phi_{W}(M,q):=\int_{M} \prod_{i=1}^{2k} \exp(\sum_{k=2}^{\infty}\frac{2}{(2k)!}G_{2k}(\tau) x^{2k}_{i})\in \mathbb{Q}[[q]],$$
where $x_{1},\cdots ,x_{2k}$ are the Chern roots  of the bundle $TM_{\mathbb{C}}=TM\otimes\mathbb{C}$, $q=e^{2\pi i\tau}$ and $G_{2k}(\tau)$ are Eisenstein series. The Witten genus $\phi_{W}(M,q)$ is a modular form of weight $2k$.

 The Witten genus was defined in \cite{witten1987elliptic}, in this work the expression below appeared considering the index of a Dirac operator on the loop space $\L M$ 
 \begin{equation}\Phi(M,q):=\int_{M} A(M)\ch (\bigotimes_{n=1}^{\infty}S_{q^{n}}(TM_{\mathbb{C}}))\prod_{n=1}^{\infty}(1-q^{n})^{\dim M}\in \mathbb{Q}[[q]].\label{WG}\end{equation}
 
It was proved in \cite{zagier1988note} that if the first Pontrjagin class is zero $p_{1}(M)=0$ on $H^{4}(M,\mathbb{R})$ the following identity is satisfied
$$\Phi(M,q)=\phi_{W}(M,q) .$$
 And if  $M$ is a spin manifold then $\Phi(M,q)\in\mathbb{Z}[[q]]$. The condition on the first Pontrjagin class was later refined in \cite{mclaughlin1992orientation}, where from a topological point of view the condition  $\frac{1}{2}p_{1}(M,\mathbb{Z})=0$ on $H^{4}(M,\mathbb{Z})$ appeared considering extensions of spin loop groups.

 In \cite{witten1987elliptic}, the author was motivated by some construction in quantum field theory \cite{alvarez1987dirac,witten1985fermion, witten1988index}. These works have motivated advances in different mathematical areas in order to provide explanations to the properties of the Witten genus and elliptic genus; some works in this direction are \cite{ taubes1989s, costello2010geometric, ando2001elliptic, han2018witten}, in particular considering vertex algebras  \cite{  gorbounov2004gerbes, ma2005elliptic, cheung2008witten, tamanoi2006elliptic, bressler2007first, huang2012meromorphic, borisov2000elliptic}.

Now we express our main results.  In this paper, we combine the geometrical constructions done in \cite{taubes1989s} with a language of operators. Following \cite{taubes1989s}, we consider first the case where $M$ is the set of fixed points of an $S^{1}$-action on a finite dimensional spin manifold $X$, we call this case the finite dimensional case. Additionally we assume that $M$ is a connected spin manifold. In this case, we define a set of objects: a $\mathbb{Z}_{2}$ graded vector bundle $\V\cong \V^{+}\oplus \V^{-}\rightarrow M$, a connection and a set of bundle endomorphisms. These objects will be used to build a global operator $Q:\Gamma(\V)\rightarrow \Gamma(\V)$ and a bundle morphism $L_{K}:\V\rightarrow \V$ such that the analytical index defined by
   \begin{equation}\Ind (Q,q):=\Tr (q^{L_{K}}|_{\ker (Q|_{\V^{+}})})-\Tr (q^{L_{K}}|_{\ker (Q|_{\V^{-}})}),\label{intro1}\end{equation}
it is going to be proved (Theorem \ref{index}) to satisfy the following identity
  \begin{equation}\Ind (Q,q)=q^{\frac{1}{2}c_{1}}\int_{M} A(M)\ch (\sqrt{\det N}\otimes_{r\in\mathcal{A}} S_{q^{r}}(N_{r})).\label{teo1}\end{equation}
  
  And we consider a second case, where $M$ is seen as the fixed point set of an $S^{1}$-action on the loop space $\L M$, additionally we assume that $M$ is a connected spin manifold. We call this case infinite dimensional case. As in the finite dimensional case,  we define a set of objects: a $\mathbb{Z}_{2}$ graded vector bundle $\V_{R}\cong \V^{+}_{R}\oplus \V^{-}_{R}\rightarrow M$, a connection and a set of bundle endomorphisms. These objects will be used to build a global operator $Q_{R}:\Gamma(\V_{R})\rightarrow \Gamma(\V_{R})$ and a bundle morphism $P:\V_{R}\rightarrow \V_{R}$ such that the index defined by
   \begin{equation}\Ind (Q_{R},q):=\Tr(q^{P}|_{\ker (Q_{R}|_{\V^{+}_{R}})})-\Tr(q^{P}|_{\ker (Q|_{\V^{-}_{R}})}),\label{intro2}\end{equation}
it is going to be proved (Theorem \ref{index2}) to satisfy the following identity
  \begin{equation}\Ind (Q_{R},q)=\frac{\Phi(M,q)}{\eta(q)^{\dim M}},\label{teo2}\end{equation}
where $\eta(q)$ is the Dedekind eta function.

  In this work, the vector bundles, connections and bundle endomorphisms are defined on the fixed point manifold of the $S^{1}$-action; this is a main difference with \cite{taubes1989s} where vector bundles, connections and bundle endomorphisms are defined on a neighborhood of the fixed point manifold. This approach in particular makes us to consider infinite dimensional vector bundles even for the finite dimensional case. Additionally this approach allows us to see a language similar to vertex algebras arising naturally. We plan to study in a future work this relation between Elliptic genus and vertex algebras.

  
The paper is organized as follows. In Section \ref{sec2}, first we fix a notation in Section \ref{sec2.1}, then we state some known geometric results Section in \ref{sec2.2}. Finally, in Section \ref{sec2.3} we give a short geometrical motivation to the Sections \ref{sec3} and \ref{sec4}.

 In Section \ref{sec3}, we consider the finite dimensional case. In Section \ref{sec3.1}, we built a $\mathbb{Z}_{2}$ graded complex bundle $\V=\V^{+}\oplus \V^{-}$ on $M$ and we define an Hermitian product on it. In Section \ref{sec3.2}, we define an Hermitian connection $D:\Gamma(\V)\rightarrow \Gamma(\V\otimes T^{*}M)$.  
 And, we define a finite set of bundles endomorphisms (Definition \ref{oscilator}) which satisfy the Heisenberg and Clifford algebra (Proposition \ref{ad1}).
In Section \ref{2.2.4} are defined the global operators, in particular, we define the global operators $Q$ and $L_{K}$ (Definition \ref{global1}). Finally in Section \ref{sec3.4}, we define the analytical index in \eqref{intro1} and it is proved in Theorem \ref{index} the equality \eqref{teo1}.

In Section \ref{sec4}, we consider the infinite dimensional case. 
This section is a generalization of the previous section.
In Section \ref{sec4.1}, we built a $\mathbb{Z}_{2}$ graded complex bundle $\V_{R}=\V_{R}^{+}\oplus \V_{R}^{-}$ on $M$ and we define an Hermitian product on it. In Section \ref{sec4.2}, we define an Hermitian connection $D:\Gamma(\V_{R})\rightarrow \Gamma(\V_{R}\otimes T^{*}M)$. And, we define an infinite set of bundles endomorphisms (Definition \ref{oscilator2}) which satisfy the Heisenberg and Clifford algebra. In Section \ref{sec4.3}, we define the global operators, in particular, we define the global operators $Q_{R}$ and $P$ (Definitions \ref{4.1} and \ref{4.2}). Finally in Section \ref{sec4.4}, it is defined the analytical index in \eqref{intro2} and it is proved in Theorem \ref{index2} the equality \eqref{teo2}.\\

\textbf{Acknowledgment}: I would like to thank Marco Aldi for his continued support while this work was done. A first version of this work was presented in the seminar Lie Group/Quantum Mathematics Seminar at Rutgers University, I am very grateful with Yi-Zhi Huang and James Lepowski by their comments and hospitality.  I am grateful to Reimundo Heluani, Bojkov Bakalov, Nicola Tarasca and Michael Penn for discussions on this and related subjects. This work was done at Virginia Commonwealth University, I am very grateful to this institution by its hospitality and excellent working conditions.

\section{Preliminaries}\label{sec2}

\subsection{Assumptions and notations}\label{sec2.1}

In this paper, $M$ denotes a $C^{\infty}$ compact connected oriented Riemannian manifold of dimension $2l$. $TM$ denotes its tangent bundle, $T^{*}M$ the cotangent bundle. $TM$ and $T^{*}M$ are $2l$-dimensional real vector bundles identified by the metric.

Let  $V\rightarrow M$ be a real, oriented $2r$-dimensional vector bundle equipped with a metric. The bundle of positively oriented, orthonormal frames $\Fr V\rightarrow M$ is a principal $SO(2r)$ bundle over $M$. The second Stiefel-Whitney class, $w_{2}(V)\in H^{2}(M,\mathbb{Z}_{2})$, is the obstruction to the existence of a principal $\Spin(2r)$ bundle, $\Fr' V\rightarrow M$ with the property that $\Fr V=\Fr' V/\{\pm 1\}$, we say that $\Fr'V$ is a spin structure. 

Let $\bigtriangleup$ be the $2^{r}$-dimensional Hermitian space of spinors, which split into $\bigtriangleup=\bigtriangleup^{+}\oplus \bigtriangleup^{-}$, where $\bigtriangleup^{+}$ and $\bigtriangleup^{-}$ are the spaces of positive and negative spinors. $\bigtriangleup^{+}$ and $\bigtriangleup^{-}$ have dimension $2^{r-1}$ and are orthogonal in $\bigtriangleup$. Moreover, $\Spin(2r)$ acts irreducibly and unitarily on $\bigtriangleup^{+}$ and $\bigtriangleup^{-}$.


When $w_{2}(V)=0$, let $\bigtriangleup(V), \bigtriangleup^{+}(V), \bigtriangleup^{-}(V)$ be the bundles of spinors over $M$
$$\bigtriangleup(V)=\Fr 'V\times_{\Spin (2r)}\bigtriangleup,\qquad \bigtriangleup^{\pm}(V)=\Fr'V\times_{\Spin (2r)}\bigtriangleup^{\pm}.$$
For every $x\in M$, the fiber $(\Fr'V)_{x}$ can be identified to a set of unitary operators from $\bigtriangleup$ into $\bigtriangleup(V)$, which send $\bigtriangleup^{\pm}$ into $\bigtriangleup(V)^{\pm}$. Let $w_{2}(M)=w_{2}(TM)$, when $w_{2}(M)=0$ we say that $M$ is a spin manifold and we denote the bundle of spinors as $\bigtriangleup(M)=\bigtriangleup(TM)$ and $\bigtriangleup^{\pm}(M)=\bigtriangleup^{\pm}(TM)$.

Let $E$ be a complex vector bundle on $M$,  we denote by $\Gamma(E)$ the space of $C^{\infty}$ complex sections. Moreover, we associate to $E$ the following vector bundles on $M$
$$S(E):=1+E+\Sym^{2}(E)+\cdots , \qquad \wedge^{*}(E):=1+E+\wedge^{2}(E)+\cdots .$$
Additionally if $m=\dim_{\mathbb{C}} E$, we associate to $E$ the following line bundle  
$$\det E= \wedge^{m}E .$$

Now, we say that a complex line bundle $L$ on $M$ has a square root if there is a  complex line bundle $S$ on $M$ such that $L=S\otimes S$, we use the notation $S=\sqrt{L}$. 

Let $X$ be a spin manifold. An orientation preserving isometric $S^{1}$-action on a spin manifold $X$ lifts to the bundle $\bigtriangleup (V)$ if there is an $S^{1}$-action on $\Fr'V$ (commuting with the action of $\Spin(2r)$ from the right $\Fr'V\times \Spin(n)\rightarrow \Fr'V$) which is compatible with the covering map $\Fr'V\rightarrow \Fr V$. 

Let $M$ be a submanifold of $X$, we say that $M$ is totally geodesic if any geodesic on the submanifold $M$ with its induced Riemannian metric is also a geodesic on the Riemannian manifold $X$.

\subsection{Geometric preliminaries}\label{sec2.2}
Let $X$ be a Riemannian manifold with an $S^{1}$-action by isometries, we denote by $K$ the Killing vector field associated to this action. We have the following general result
\begin{proposition}\label{pro2.1} The zero set of a Killing field is a disjoint union of totally geodesic submanifolds each of even codimension.
\end{proposition}
 See \cite[Theorem II.5.1]{kobayashi1995transformation} 
 for a proof of this proposition.

\begin{theorem}\label{Koba} Let $M$ be the set of fixed points for the Killing vector field $K$ on a Riemannian manifold $X$, and let	 $N'=TM^{\perp}$ be the normal bundle of $M$ in $X$. Then

\emph{(i)}\; Considered as a linear endomorphims of $TX$, the covariant derivative $\nabla^{X} K$ annihilates the tangent bundle $TM$ and induces a (skew-symmetric) automorphism of the normal bundle $N'$. Restricted to $N$, $\nabla^{X} K$ is parallel, $\nabla^{X}_{u}(\nabla^{X} X)$ for every $u\in TM$. 

\emph{(ii)}\; The normal bundle $N'$ is a complex vector bundle i.e. it exists an endomorphism called {complex structure} $J:N'\rightarrow N'$ s.t $J^{2}=-Id$.

\emph{(iii)}\; If $M$ is orientable, then $N'$ is orientable.

\end{theorem}
\begin{proof} 
We follow the proof in \cite[Theorem II.5.3]{kobayashi1995transformation}. We denote by $\nabla^{X}$ the Levi-Civita connection on $X$. Let $K$ be the Killing vector field associated to the $S^{1}$-action on $X$. We have the linear map 
 $$(\nabla^{X} K):TX\rightarrow TX, \qquad u\mapsto (\nabla^{X}_{u} K), $$
 where $u\in TM$.  Then by definition of $M$ we have that $T_{x}M = \ker (\nabla^{X} K)_{x}$ and $N'_{x}=\Ima(\nabla^{X} K)_{x}$ for $x\in M$. Additionally, Killing vector fields satisfy the Killing equation, hence $\nabla^{X} K$ defines a skew-symmetric map on $TM$.  Then,  for each $x\in M$ there is a basis of $T_{x}X= T_{x}M\oplus N'_{x}$ such that the linear endomorphism $(\nabla^{X} K)_{x}$ is given by a matrix of the form
 $$\begin{bmatrix}0 &&& & & & &\\
                                  & \cdots&&& & && \\
                                  & & 0 && & &&\\
                                  && & 0& a_{1}&& &\\
                                  &&&-a_{1} &0 &&&\\
                                  &&& & &\cdots&&\\
                                  &&& & && 0& a_{n} \\
                                  &&&&&&-a_{n} &0.  \end{bmatrix}, \quad\quad a_{i}> 0 ,$$
 where $2n=\dim X-\dim M$.
 
Now, Killing vector fields satisfy the following general equation 
$$\nabla^{X}_{u}\nabla^{X} K=-R(K,u), \qquad \text{for every }u\in TX , $$
where $R$ is the tensor curvature.  At every point of $M$, the right hand side of $R(K,u)$ vanishes and hence $\nabla^{X}_{v}\nabla^{X} K=0$ for every vector tangent to $X$ at a point of $M$, see the reference for details. 

Since, the eigenvalues $\pm i a_{1},\cdots, \pm i a_{n}$ of $(\nabla^{X} K)_{x}$ defined above remain constant on $M$ because $\nabla^{X} K$ is parallel on $M$, we can decompose the normal bundle $N'$ in to sub-bundles $N'_{1},\cdots N'_{k}$ as follows
$$N'=N'_{1}+\cdots +N'_{k},\qquad (\text{orthogonal decomposition}), $$
where $N'_{1},\cdots N'_{k}$ correspond to the eigenspaces for $\pm i r_{1},\cdots ,\pm i r_{k}$ of $(\nabla^{X} K)_{x}$ restricted to $N'$. We use the notation $a_{i}$ for possibly equal values and $r_{i}$ for different values.  Let $J$ be the endomorphism of $N'$ defined by $J|_{N_{r}}:=\frac{1}{r}(\nabla^{X} K)$. Then, $J^{2}=-Id$, and $J$ defines a complex vector bundle structure in $N'$. Let $J:N'\rightarrow N'$ be the endomorphism defined by $J|_{N_{r}}:=\frac{1}{r}(\nabla^{X} K)$, then $J^{2}=-Id$. Since, $TX_{M}=TM+N'$, if $TM$ is orientable, $TN$ is orientable.
  \end{proof}
  In the previous theorem, $K$ is a general Killing vector field. We have that the flow of $K$ is given by an $S^{1}$-action, hence $a_{1},\cdots, a_{n}$ are positive integers. From now on, we denote this finite set of integers as follows 
   \begin{equation}\label{set}
   \A=\{r_{1}, \cdots, r_{k}\},
   \end{equation}
   where as before we use the notation $a_{i}$ for possibly equal values and $r_{i}$ for different values.

 From Theorem \ref{Koba} (ii), $N'$ is a real vector bundle on $M$ with a complex structure $J$, hence $N'$ can be turned into a complex vector bundle, denoted by $(N',J)$,  by setting: For any point $x\in M$, any real numbers $a, b\in\mathbb{R}$ and vector $u\in N'_{x}$,
$$(a+ib)u:=au+bJu .$$
Moreover, the endomorphism $J$ is linearly extended to $N'\otimes\mathbb{C}$, and the eigenvalues $i$ and $-i$ of $J$ splits $N'\otimes\mathbb{C}$ in complex bundles $N$ and $\bar{N}$ respectively as follows
\begin{equation}\label{ult}
N'\otimes\mathbb{C}=N+\bar{N}.
\end{equation}
There is a complex isomorphism between $(N',J)$ and $N$ given by $u\rightarrow u-iJu$ for $u\in N'$.

\begin{corollary}\label{proJ}
The complex structure $J$ on $N'$ is orthogonal and parallel i.e. $J_{p}\in SO(2n,\mathbb{R})$ for every $p\in M$ and $\nabla_{u}J=0$ for every $u\in TM$. 
\end{corollary} 

This corollary follows directly from the definition of $J$ in Theorem \ref{Koba}. 
Aditionally, we have an Hermitian extension of the Riemannian metric $g_{X}$ on $X$ to the bundle $N'\otimes\mathbb{C}=N+ \bar{N}$ defined by 
\begin{equation}
(u\otimes a, v\otimes b):=a\bar{b}g_{X}(u,v)\label{g}
\end{equation}
where $u,v\in N, a,b\in\mathbb{C}$. The bundles $N$ and $\bar{N}$ form an orthogonal decomposition with respect $(\cdot, \cdot)$. Hence, $N$ and $\bar{N}$ are Hermitian bundles, the Hermitian metric on $N$ is defined by $({n},{m})_{N}:=(n,m)$ where $n,m\in\Gamma(N)$, and the Hermitian metric on $\bar{N}$ is given by $(\bar{n},\bar{m})_{\bar{N}}:=({m},{n})$ where $\bar{n},\bar{m}\in\Gamma(\bar{N})$.

Now, the complex bundles $N$ and $\bar{N}$ have the following decomposition. 

\begin{corollary}\label{Koba2} The operator $\nabla^{X} K:N'\rightarrow N'$ splits the normal bundle as follows 
$$ N'\otimes \mathbb{C}=\bigoplus_{r\in \mathcal{A}}N_{r}+ \bar{N}_{r},\quad\quad N_{r}:=\{v\in N'\otimes\mathbb{C}| \nabla^{X}_{v} K=-irv\},$$
 Therefore, we have isomorphisms $N=\bigoplus_{r\in \mathcal{A}}N_{r}$ and $\bar{N}=\bigoplus_{r\in \mathcal{A}}\bar{N}_{r}$.
\end{corollary}
This corollary is obtained from proof of Theorem \ref{Koba}. As we mention in the introduction, we denote the dimensions of these bundles by $\dim_{\CC}N_{r}=d_{r}$.

The complex line bundles on $M$ form a group with respect to tensor product. The first Chern class defines an isomorphism between this group and $H^{2}(M,\mathbb{Z})$; hence a line bundle $L$ has square root iff $c_{1}(L)=0$ (mod $2$).

\begin{proposition}\label{pro2.2} Let $X$ be a spin manifold such that the $S^{1}$-action on $X$ lifts to the spin bundle $\bigtriangleup(X)$ and $M$ is the fixed point set of the $S^{1}$ action on $X$. We have that $M$ is a spin manifold iff $\sqrt{\det N}$ exists. 
\end{proposition}

\begin{proof}  From the restriction of the spin bundle $\Fr'_{X}|_{M}$ we have that the bundle $TX_{M}=TM+N'$ on $M$ has second Stieffel-Whitney class $w_{2}(TM+N')=0$ on $H^{2}(M,\mathbb{Z}_{2})$. Note that $w_{2}$ of a direct sum is the sum of the $w_{2}$'s from each summand when the summands are oriented. Also, for complex bundles we have that $w_{2}(N)=c_{1}(N)$ (mod $2$) and $c_{1}(N)=c_{1}(\det N)$. 
\end{proof}

Additionally, we have the following proposition.  

\begin{proposition}\label{spin}Given three real vector bundles $V_{1}, V_{2}$ and $V=V_{1}+V_{2}$ over a manifold $M$,a choice of orientation on any two of them uniquely determines an orientation on the third. Similarly, a choice of spin structure on any of them uniquely determines a spin structure on the third.
\end{proposition}
See \cite[II Proposition 1.5]{lawson2016spin} for a proof of this proposition. Let us assume that $\bigtriangleup(V_{1})$ and $\bigtriangleup(V_{2})$ are spin bundles on $M$, then the unique spin structure $\Fr' V$ on $M$ in the previous proposition satisfies that  
\begin{equation}\label{1}
\bigtriangleup(V)=\bigtriangleup(V_{1})\otimes\bigtriangleup(V_{2}), 
\end{equation}
where $\otimes$ is the graded tensor product. Throughtout this work $\otimes$ denotes the $\ZZ_{2}$ graded tensor product. Finally, let $V'$ be a real bundle on $M$ with complex structure $J$, then as before we have that $V'\otimes \CC=V+\bar{V}$. Additionally, if $V'$ defines an spin structure on $M$ one has  
\begin{equation}\label{2}
\bigtriangleup(V')=\sqrt{\det V}\otimes \wedge^{*}\bar{V},
\end{equation}
where $\sqrt{\det V}$ exists and is uniquely defined by the spin structure.

Finally, we have the following proposition.

\begin{proposition}\label{pro2.6} Let $M$ be a submanifold of a Riemannian manifold $X$, we denote by $\nabla^{X}$ the Levi-Civita connection on $X$ and by $\nabla^{M}$ the Levi-Civita connection that comes from the metric induced by $g_{X}$. We define the second fundamental form $II:TM\times TM\rightarrow N'$ of $M$ in $X$ by 
$$II(u,v)=\nabla^{X}_{u}v-\nabla^{M}_{u}v, $$
 for $u,v\in TM$. One has that $II=0$ iff $M$ is totally geodesic.
\end{proposition}
See \cite[VI Proposition 2.6]{do1992riemannian}. Now, we have for an arbitrary $n\in \Gamma(N')$ that
$$0=g_{X}(II(u,v), n)=g_{X}(\nabla^{X}_{u}v, n)=-g_{X}(v, \nabla^{X}_{u}n).$$
Therefore,  $\nabla^{X}$ defines on $M$ a connection on the normal bundle 
\begin{equation}\label{eq2.4}
\nabla^{X}:\Gamma(N')\rightarrow \Gamma(N'\otimes TM),
\end{equation}
we keep the same notation $\nabla^{X}$ by abuse of language.

\subsection{Motivation} In this subsection, we give a short motivation for the constructions in Section \ref{sec3} and \ref{sec4}, we refer to \cite{taubes1989s, witten1985fermion, deligne1999quantum} for details.

\label{sec2.3} The identity \eqref{ASS} asserts that the equivariant index can be computed from geometric data at the fixed point set $M$ of the $S^{1}$-action. This interpretation is formalized considering a modification of the Dirac operator $\slashed{D}_{X}$ as follows
\begin{equation}\label{eq2.16}
Q_{t}:=\slashed{D}_{X}+itK': \Gamma(\bigtriangleup(X))\rightarrow \Gamma(\bigtriangleup(X)\otimes T^{*}X),
\end{equation}
where $K'\in T^{*}M$, the dual of $K\in TM$, acts on $\bigtriangleup(X)$ by Clifford multiplication. Additionally, we consider the Lie derivative of the $S^{1}$ action on the spinors $\bigtriangleup(X)$ which is given by the operator below 
\begin{equation}\label{eq2.17}
 iL_{K}:=D^{X}_{K}-\frac{1}{4}dK': \Gamma(\bigtriangleup(X))\rightarrow \Gamma(\bigtriangleup(X)),
 \end{equation}
where $D^{X}$ is the connection on $\bigtriangleup(X)$ associated to the Levi-Civita connection and $dK'$ is a 2-form which acts on  $\bigtriangleup(X)$ by Clifford multiplication. 

Now, the analytic equivariant index for $Q_{t}$ is defined, as in \eqref{equiin}, as follows
 \begin{equation}\label{qin}
 \Ind (Q_{t},q)=\Tr (q|_{\ker (Q_{t}|_{\bigtriangleup^{+}(X)})})-\Tr (q|_{\ker (Q_{t}|_{\bigtriangleup^{-}(X)})}),
 \end{equation}
and we have the identity
\begin{equation*}
\Ind(Q_{t},q)=\Ind(\slashed{D}_{X},q).
 \end{equation*}
 Hence, the index in \eqref{qin} is independent of $t$. The Weitzenb\"{o}ck formula for $Q_{t}$ is given by 
 $$Q^{2}_{t}=(\slashed{D}_{X})^{2}+t^{2}g_{X}(K,K)-2tL_{K}+it\frac{1}{2}dK'.$$
 Therefore, as $t\rightarrow \infty$ the support of the sections in $\ker (Q_{t}|_{\bigtriangleup^{+}(X)})$ and $\ker (Q_{t}|_{\bigtriangleup^{-}(X)})$ must shrink around the points on $X$ where $g_{X}(K,K)=0$, and by definition of $K$ we have that the set of points where $g_{X}(K,K)=0$ is given by the fixed point set $M$ of the $S^{1}$-action. Then, we can compute the equivariant index working on the space of sections 
 \begin{equation}\label{gamma}\Gamma(\bigtriangleup(X)|_{\U}),\end{equation}
where $\U$ is a tubular neighborhood of $M$ in $X$. 
This approach in particular leads to a proof of the identity \eqref{ASS} and the Atiyah-Hirzebruch theorem, see the references mentioned at the beginning of this subsection.

\section{Localization, finite dimensional case}\label{sec3}

\subsection{The Vector Bundle $\V$}\label{sec3.1}

From now on, we assume that $X$ is an oriented, compact, $2m$-dimensional, spin Riemannian manifold with an $S^{1}$-action which lifts to the spin bundle $\bigtriangleup(X)$, and $M$ is the fixed point set of the $S^{1}$ action on $X$. Additionally, we assume that $M$ is a  $2l$-dimensional, connected, spin manifold. 

Now, we define the vector bundle $\V\rightarrow M$, and we give a motivation to the definition of this bundle. From Theorem \ref{Koba} and Proposition \ref{pro2.2}, we build the following vector bundle
$$\V:=\bigtriangleup (M)\otimes \sqrt{\det N}\otimes \wedge^{*} (\bar{N})\otimes S(N)\otimes S(\bar{N}) \rightarrow M.$$
The definition of the bundle $\V$ is motivated by the space of sections $\Gamma(\bigtriangleup(X)|_{\U})$ defined in \eqref{gamma}, we describe now the relation between these spaces: 

First, we have that $TX_{M}=TM+N'$ on $M$, then we have the following isomorphism of bundles on $M$ 
$$\bigtriangleup(X)=\bigtriangleup(TM+N')=\bigtriangleup(M)\otimes \sqrt{\det N}\otimes \wedge^{*} (\bar{N}). $$
where we use  \eqref{1} and \eqref{2}.

And second, the bundles $S(N)\otimes S(\bar{N})$ are motivated by the coordinates in the normal direction of $M$ in $\U$. For any $p\in M$ there is a local chart $ U$ of $\U$ where $p\in U$ and we have a diffeomorphism
$$(x_{1},\dots, x_{2l}, y_{1}, \bar{y}_{1}, \cdots y_{n}, \bar{y}_{n})\in U'\subset \R^{2m} \rightarrow U\subset \U ,$$
such that 
$$(x_{1},\dots, x_{2l}, 0, \cdots , 0)\in U'\subset \R^{2m} \rightarrow  M \cap U\subset \U. $$
Notice that sections of $\Gamma(\bigtriangleup(X)|_{U})$ does depend of $y_{1}, \bar{y}_{1}, \cdots y_{n}, \bar{y}_{n}$. On the other hand, sections of $\Gamma(\bigtriangleup(X)|_{M\cap U})$ does not depend of $y_{1}, \bar{y}_{1}, \cdots y_{n}, \bar{y}_{n}$, hence we add this to space of sections as follows 
$$\Gamma(\bigtriangleup(X)|_{M\cap U})\otimes \mathbb{C}[ y_{1}, \bar{y}_{1}, \cdots y_{n}, \bar{y}_{n}].$$ 
This local relation motivates us to consider globally the vector bundle $S(N)\otimes S(\bar{N})$. Hence, we are leaded to define $\V$ as the tensor product of the bundles $\bigtriangleup(M)\otimes \sqrt{\det N}\otimes \wedge^{*} (\bar{N})$ and $S(N)\otimes S(\bar{N})$ on $M$.  

\begin{remark} The idea of working with a bundle $\V$ on $M$ instead of $\bigtriangleup(X)|_{\U}$ simplifies the constructions  in the infinite dimensional case, because we work with vector bundles on $M$ without making any consideration on the loop space $\L M$. 
\end{remark}


Now,   $\bigtriangleup (M)=\bigtriangleup^{+} (M)+\bigtriangleup^{-} (M)$ is $\mathbb{Z}_{2}$ graded and $\wedge^{*}\bar{N}=\wedge^{\even}\bar{N}+\wedge^{\odd}\bar{N}$ is $\mathbb{Z}_{2}$ graded. Then, $\V=\V^{+}+\V^{-}$ is a $\mathbb{Z}_{2}$ graded vector bundle considering the graded tensor product   
\begin{align*}
&\V^{+}=\left( \bigtriangleup^{+}(M)\otimes\wedge^{\even}\bar{N}+ \bigtriangleup^{-}(M) \otimes\wedge^{\odd}\bar{N} \right)\otimes \sqrt{\det N}\otimes S(N)\otimes S(\bar{N}),\\
&\V^{-}=\left( \bigtriangleup^{-}(M)\otimes\wedge^{\even}\bar{N}+ \bigtriangleup^{+}(M)\otimes\wedge^{\odd}\bar{N} \right)\otimes \sqrt{\det N}\otimes S(N)\otimes S(\bar{N})\label{v-}.
\end{align*}

Moreover, we have from Proposition \ref{Koba2} that
$$\V=\bigtriangleup (M)\otimes \sqrt{\det N}\otimes_{r\in\mathcal{A}} \wedge^{*}(\bar{N}_{r})\otimes_{r\in \mathcal{A}} S(N_{r})\otimes_{r\in\mathcal{A}} S (\bar{N}_{r}) \rightarrow M.$$

Observe that this decomposition respects the $\mathbb{Z}_{2}$ grading on $\V$.  

\subsubsection{Hermitian product}\label{IP}

The vector bundle $\V$ is defined as a tensor product of bundles, hence we define an Hermitian product on $\V$ considering Hermitian products for each one of the bundles which define $\V$. 

First, $M$ is spin manifold and its spin bundle $\bigtriangleup(M)$ has an Hermitian product, see \cite{lawson2016spin} for details. 

Second,  we use the Hermitian product defined in \eqref{g} on $N$ to induce an Hermitian product on $T(N):=\sum_{l\geq0}N^{\otimes l}$ 
\begin{equation}
(n_{1}\otimes \cdots \otimes n_{k},m_{1}\otimes \cdots\otimes m_{l})_{T(N)}=\delta_{k,l}(n_{1},m_{1})\cdots(n_{k},m_{k}),
\label{Inner}\end{equation}
Where $ n_{1}\otimes \cdots \otimes n_{k}, m_{1}\otimes \cdots \otimes m_{l}\in T(N)$ and $k,l\in\ZZ_{+}$. Therefore, we have Hermitian  products on the subspaces $S(N)$ and $\wedge^{*}N$ given by 
\begin{equation}(n_{1}\cdots n_{k}, m_{1}\cdots m_{l})_{S(N)}=\delta_{k,l}\sum_{\sigma\in \mathcal{S}_{k}} (n_{\sigma(1)},m_{1})\cdots(n_{\sigma(k)},m_{k}),\label{Innersym}\end{equation}
\begin{equation}(n_{1}\wedge\cdots\wedge n_{k}, m_{1}\wedge\cdots\wedge m_{l})_{\wedge N}=\delta_{k,l}\sum_{\sigma\in \mathcal{S}_{k}}\epsilon(\sigma) (n_{\sigma(1)},m_{1})\cdots(n_{\sigma(k)},m_{k}),\label{Innerwedge}\end{equation}
where $ n_{1} \cdots n_{k}, m_{1} \cdots m_{l}\in S(N)$, $ n_{1}\wedge \cdots\wedge n_{k}, m_{1}\wedge \cdots\wedge m_{l}\in \wedge^{*}N$, $\mathcal{S}_{n}$ is the group of permutations and $\epsilon:\mathcal{S}_{n}\rightarrow \{1,-1\}$ the sign homomorphism. 

Third, $\det (N)$ is a subspace of $\wedge^{*}N$ and its Hermitian  product is defined by \eqref{Innerwedge}. The definitions of Hermitian products of $S(\bar{N})$ and $\wedge^{*}\bar{N}$ follow similarly.

Four, for $S(N)\otimes S(\bar{N})$ we do not use the Hermitian structure obtained from the tensor product of the Hermitian bundles $S(N)$ and $S(\bar{N})$. We consider a different Hermitian product motivated by the Harmonic oscillator. In the next subsection, we are going to define bundles maps in $\End (\V)$ (Definition \ref{oscilator}) which satisfy the algebra of the Harmonic oscillators (Proposition \ref{ad1}),  the inner product defined below makes these operators have the usual conjugate operators  (Proposition \ref{ad}). 

The Hermitian product of $S(N_{r})\otimes S(\bar{N}_{r})$ is defined as follows. We restrict the Hermitian product in \eqref{g} on the subbundles $N_{r}$. Let $a_{1}\cdots . a_{k}\bar{b}_{1}\cdots . \bar{b}_{s}, c_{1}\cdots . c_{l}\bar{d}_{1}\cdots . .\bar{d}_{t}\in \Gamma(S(N_{r})\otimes S(\bar{N}_{r}))$, we relabel these sections as $n_{1}\cdots . n_{k}\bar{m}_{l+1}\cdots . \bar{m}_{l+s}$ and $m_{1}\cdots . m_{l}\bar{n}_{k+1}\cdots . \bar{n}_{k+t}$ respectively; the Hermitian product is given by
\begin{equation}
\begin{split}(n_{1}\cdots n_{k}\bar{m}_{l+1}&\cdots  \bar{m}_{l+s}, m_{1}\cdots m_{l}\bar{n}_{k+1}\cdots  \bar{n}_{k+t}):=\frac{1}{r^{k+t}}(n_{1}\cdots  n_{k+t},m_{1}\cdots m_{l+s})\\
&=\frac{1}{r^{k+t}}\delta_{k+t,l+s}\sum_{\sigma\in \mathcal{S}_{k+r}} (n_{\sigma(1)},m_{1})\cdots (n_{\sigma(k)},m_{l})(n_{\sigma(k+1)},l_{1})\cdots (n_{\sigma(k+r)},l_{s}).\label{Innerss}
\end{split}\end{equation}
 The restriction of this Hermitian product on the subspaces $S(N_{r})$ and $S(\bar{N}_{r})$ is different to \eqref{Innersym}, the difference comes from the factors $\frac{1}{r^{k}}$ where $k\in \mathbb{Z}_{+}$. We build the Hermitian product on $S(N)\otimes S(\bar{N})$ from the tensor product of the Hermitian bundles $S(N)\otimes S(\bar{N})=\otimes_{r\in\mathcal{A}} S(N_{r})\otimes S(\bar{N}_{r})$.

Finally, the Hermitian product on $\V$ is given by the tensor product of the inner products on $\bigtriangleup(M)$, $\det N$, $\wedge^{*} \bar{N}$ and $S(N)\otimes S(\bar{N})$. Hence $\V$ is an Hermitian bundle, we denote this Hermitian product by $(\cdot, \cdot)$.

\subsection{Operators }\label{sec3.2}

In this section, we define the operators acting on the bundle $\V$. In this work, we call operators indistinctly either connections or bundles morphisms


\subsubsection{Connections}\label{sec3.2.1}

Now, we build an Hermitian connection on $\V$, once again we use that $\V$ is defined as a tensor product in terms of other bundles, hence we obtain an Hermitian connection on $\V$ considering Hermitian connections for each one of the bundles which define $\V$.

First, the spin manifold $M$ has a metric induced from the Riemannian metric on $X$, hence, we have a Levi-Civita connection $\nabla^{M}$ on $M$.  From the Levi Civita connection we obtain a connection on the space of spinors $D^{\bigtriangleup}:\Gamma(\bigtriangleup (M))\rightarrow \Gamma(\bigtriangleup (M)\otimes   T^{*}M)$, this connection is compatible with the Hermitian metric on $\bigtriangleup (M)$, see \cite{lawson2016spin}.

Second, from \eqref{eq2.4}, we have that $\nabla^{X}$ defines on $M$ a connection on the normal bundle $N'$. Additionally, from Corollary \ref{proJ} we have that $\nabla^{X}$ defines an Hermitian complex connection on $(N'J)$ (see \eqref{ult}) or $N$.  And, from Theorem \ref{Koba} (i) and Corollary \ref{Koba2} we obtain the following Hermitian complex connections 
\begin{equation}\label{con}
\begin{split}
&\nabla^{X}:\Gamma(N_{r})\rightarrow \Gamma(N_{r}\otimes   T^{*}M),\qquad \nabla^{X}:\Gamma(\bar{N}_{r})\rightarrow \Gamma(\bar{N}_{r}\otimes   T^{*}M),
\end{split}
\end{equation}
where $r\in \A$.


Third, the connections in \eqref{con} induce connections on the tensor products $T(N_{r})$ and $T(\bar{N}_{r})$ , these connections are compatible with the Hermitian product $\eqref{Inner}$. Hence, from \eqref{con} we define the following natural  connections on  $S( N_{r})$, $S( \bar{N}_{r})$, $\wedge^{*}(\bar{N}_{r})$ for $r\in \A$ and $\det N$
\begin{equation}  \label{con2}
\begin{split}
&D^{S}:  \Gamma(S(N_{r}))\rightarrow \Gamma(S( N_{r})\otimes   T^{*}M ),\\
&D^{\bar{S}}: \Gamma(S( \bar{N}_{r}))\rightarrow \Gamma(S(\bar{N}_{r})\otimes   T^{*}M),\\
&D^{\wedge^{*}}:\Gamma(\wedge^{*}(\bar{N}_{r}))\rightarrow \Gamma(\wedge^{*}(\bar{N}_{r})\otimes   T^{*}M),\\
& D^{\det}:\Gamma( \det N)\rightarrow \Gamma(\det N\otimes   T^{*}M).
\end{split}
\end{equation}
The compatibility of \eqref{con} implies the compatibility of the connections above with respect to their Hermitian products, we prove the only nontrivial compatibility in the following proposition  

\begin{proposition} \label{sscone}The connection $ D^{S}\otimes  D^{\bar{S}}: \Gamma(S(N)\otimes S(\bar{N}))\rightarrow \Gamma(S(N)\otimes S(\bar{N})\otimes T^{*}M)$ is compatible with the Hermitian  metric \eqref{Innerss}.
\end{proposition}
\begin{proof}Let $k,t,l,s$ positive integers.  If $k+t=l+s$ we have from a direct calculation that 
 \begin{align*}&( D^{S}\otimes  D^{\bar{S}}(n_{1}\cdots n_{k}\bar{m}_{l+1}\cdots \bar{m}_{l+s}), m_{1}\cdots m_{l}\bar{n}_{k+1}\cdots \bar{n}_{k+t})\\
 &\qquad+(n_{1}\cdots n_{k}\bar{m}_{l+1}\cdots \bar{m}_{l+s},  D^{S}\otimes  D^{\bar{S}}(m_{1}\cdots m_{l}\bar{n}_{k+1}\cdots\bar{n}_{k+t}))\\
 &=( D^{S}(n_{1}\cdots n_{k+t}),m_{1}\cdots m_{l+s})+(n_{1}\cdots n_{k+t}, D^{S}(m_{1}\cdots m_{l+s}))\\
 &=d(n_{1}\cdots n_{k+t},m_{1}\cdots m_{l+s})=d(n_{1}\cdots n_{k}\bar{m}_{l+1}\cdots \bar{m}_{l+s}, m_{1}\cdots m_{l}\bar{n}_{k+1}\cdots \bar{n}_{k+t})
 \end{align*}
 where $n_{1}\cdots n_{k}\bar{m}_{l+1}\cdots \bar{m}_{l+s}, m_{1}\cdots m_{l}\bar{n}_{k+1}\cdots \bar{n}_{k+t}\in S(N)\otimes S(\bar{N})$. And, we have that the case $k+t\neq l+s$ is trivial.
\end{proof}

Finally, we obtain the following connection on $\V$
\begin{equation}D^{M}:=D^{\bigtriangleup}\otimes D^{det}\otimes D^{\wedge^{*}}\otimes D^{S}\otimes  D^{\bar{S}}:\Gamma(\V)\rightarrow \Gamma(\V\otimes T^{*}M).\label{connection}\end{equation}
This connection is compatible with the Hermitian metric $(\cdot, \cdot)$ on $\V$.

\subsubsection{ Sheaves of operators I}\label{sheope}
In this subsection, we define the operators on the fibers of the bundle $\V\rightarrow M$. Let $E_1$ and $E_2$ complex vector bundles on $M$, we denote by $\Hom (E_1, E_2)$ the complex vector bundle of $\CC$-linear maps between $E_1$ and $E_2$. Additionally, if $E$ is another complex vector bundle we denote a $\CC$-linear map between $E$ and $\Hom (E_1, E_2)$ as follows
 $$E\rightarrow \Hom (E_1, E_2)$$
  If $E_1=E_2=E$ then we use the notation $\End(E)=\Hom(E,E)$ and a $\CC$-linear map is denoted by
 $$E\rightarrow\End(E).$$ 

\begin{definition} \label{def1}
Let $x\in M$ an arbitrary point on $M$ and let $n\in (N_{r})_{x}$, $\bar{m}\in (\bar{N}_{r})_{x}$,  $n_{1}\cdots n_{i}\cdots n_{k}\in S(N_{r})_{x}$ and $n_{1}\wedge\cdots n_{i} \cdots \wedge n_{k}\in (\wedge^{*} N_{r})_{x}$ for $k\in \ZZ_{+}$ complex vectors in the fibers of the complex bundles $N_{r}, \bar{N}_{r}, S(N_{r}), \wedge^{*} N_{r}$ respectively. We define the following natural set of $\mathbb{C}$-linear maps
\begin{align*} 
&\cdot_{r}: =N_{r}\rightarrow \End (S( N_{r})),\quad\quad n\cdot (n_{1}\cdots n_{i}\cdots n_{k}):=n n_{1}\cdots n_{i}\cdots n_{k}, \\[5pt]
&\wedge_{r} :=\bar{N}_{r}\rightarrow \End (\wedge^{*} \bar{N}_{r}),\quad\quad \bar{n} \wedge (\bar{n}_{1}\wedge\cdots  \bar{n}_{i}\cdots \wedge \bar{n}_{k}):=\bar{n}\wedge \bar{n}_{1}\wedge\cdots  \bar{n}_{i}\cdots  \wedge \bar{n}_{k} ,\\[-4pt]
&\partial'_{r}: =\bar{N_{r}}\rightarrow \End (S( N_{r})),\quad\quad  \bar{m}\mapsto\partial'_{{m}}(n_{1}\cdots n_{i}\cdots  n_{k}):=\sum^{k}_{i=1} (n_{i},m)n_{1}\cdots \hat{n}_{i}\cdots  n_{k},\\[-10pt]
&i_{r}:={N}_{r}\rightarrow \End (\wedge^{*} \bar{N}_{r}),\quad\quad  {m}\mapsto i_{{\bar{m}}} (\bar{n}_{1}\wedge\cdots \bar{n}_{i}\cdots\wedge \bar{n}_{k}):=\sum^{k}_{i=1}(-1)^{i+1}(\bar{n}_{i},\bar{m})\bar{n}_{1}\wedge\cdots \hat{\bar{n}}_{i}\cdots \wedge \bar{n}_{k},
\end{align*}

\vspace{-0.2cm}

in lines two, three and four we consider the complex conjugate map ${N}_{r}\rightarrow \bar{N}_{r}$ and the Hermitian metric \eqref{g} on $N_{r}$ and $\bar{N}_{r}$.
\end{definition}
 Using the transition functions for the Hermitian bundles above it is easy to see that these maps are well defined, below we show this for the map $\partial'_{r}$. 
Let $U$ and $U'$ open subsets of $M$ such that  $g:U\cap U'\rightarrow  U(d_{r})$ is a transition function of the bundle $N_{r}$ then we have
{\small$$\xymatrix{ (\bar{n},n_{1}\cdots n_{i}\cdots n_{k})\in  \Gamma(\bar{N}\times S(N_{r}))|_{U\cap U'} \ar[d]^g \ar[r]^{\partial'} &\sum_{i=1}^{k}(n_{i},n)n_{{1}}\cdots \hat{n}_{i}\cdots n_{{k}}\in   \Gamma (S(N_{r}))|_{U\cap U'}\ar[d]^g \\
(\bar{g}\bar{ n}, (gn_{{1}})\cdots (gn_{{k}}))\in  \Gamma(\bar{N}\times S(N_{r}))|_{U\cap U'} \ar[r]^{}    & \sum_{i=1}^{k}(gn_{i},g n)(gn_{{1}})\cdots \hat{g}\hat{n}_{i} \cdots (gn_{{k}})\in  \Gamma( S(N_{r}))|_{U\cap U'}}$$}
where $n\in \Gamma(N_{r})|_{U\cap U'}$, $\bar{n}\in \Gamma(\bar{N}_{r})|_{U\cap U'}$ and $n_{1}\cdots n_{i}\cdots n_{k}\in \Gamma(S(N_{r}))|_{U\cap U'}$, by abuse of language we use the same notation in Definition \ref{def1} for vectors on the fibers. 

Let $E$ and $E'$ be complex vector bundles on $M$, and $a$ a $\CC$-linear map $E\rightarrow \End(E)$. We associate a map on  $E\rightarrow \End(E\otimes E')$ to $a$ as follows: Let $x\in M$ an arbitrary point  on $M$ and $a_{x}$ and $e$ vectors in $E_{x}\rightarrow \End(E)_{x}$ and $E_{x}$ respectively then
\[a_{x} :e\mapsto a(e)\otimes  Id_{E'}\in \End(E\otimes E').\]
We call this map an extension of the image of $a$, by abuse of language we denote this map by $a$. In particular, we consider the extensions of the bundles maps in Definition \ref{def1} on $\V$, hence we have $\CC$-linear maps $\cdot_{r}:N_{r}\rightarrow \End (\V)$, $\wedge_{r} :\bar{N}_{r}\rightarrow \End (\V)$, $\partial'_{r}: \bar{N_{r}}\rightarrow \End (\V)$ and $i_{r}:{N}_{r}\rightarrow \End (\V)$.
From now on, we use the following notation
\begin{align*}
&\cdot_{r}(n)={n}\cdot \in \End(\V)\,, \qquad n\in \Gamma(N_{r}),\\
&\wedge_{r}(\bar{n})=\bar{n}\wedge\in\End(\V)\,, \quad\,  \bar{n}\in \Gamma(\bar{N}_{r}),\\
&\partial'_{r}(\bar{n})=\partial'_{n}\in \End(\V)\,,\qquad \bar{n}\in \Gamma(\bar{N}_{r}),\\
&i_{r}({n})=i_{\bar{n}}\in \End(V)\, ,\qquad {n}\in \Gamma({N}_{r}).
\end{align*}

Now, the $\mathbb{Z}_{2}$ grading of the vector bundle $\V=\V^{+}+\V^{-}$ induces a $\mathbb{Z}_{2}$ grading on $\End(\V)=\End(\V)^{+}+\End(\V)^{-}$. Moreover $\End(\V)$ forms an associative algebra with product given by the composition, and this product respects the $\ZZ_{2}$ grading. In particular for arbitrary $n\in \Gamma(N_{r})$ and $\bar{n}\in \Gamma(\bar{N}_{r})$ we have that $n\cdot, \partial'_{n}\in \End(\V)^{+}$ and $\bar{n} \wedge , i_{\bar{n}}\in  \End(\V)^{-}$.

Additionally, we define a graded commutator on $\End(\V)$ as follows 
$$[a,b]=ab-(-1)^{p(a)p(b)}ba, \qquad a,b\in \End(\V),$$
where $p$ denotes the parity;  $p(a)=0$ if $a\in \End(\V)^{+}$ and $p(a)=1$ if $a\in \End(\V)^{-}$. 

\begin{lemma}\label{lemmun} The operators $\cdot_{r}, \wedge_{r}:N_{r}\rightarrow \End (\V)$, $\partial'_{r}, i_{r}:\bar{N}_{r}\rightarrow \End (\V)$ satisfy the algebra below with respect the the graded commutator
$$[\partial'_{m},n]=(n,m), \quad  [i_{\bar{m}},\bar{n}\wedge]=(\bar{n},\bar{m}), \qquad  n, m  \in \Gamma(N_{r}).$$
\end{lemma}
The proof of this lemma follows directly from Definition \ref{def1}.  In the following lemma, we consider the extension of the operators $\cdot_{r}$ and $\partial'_{r}$ in Definition  \ref{def1} to the bundle $S(N)$.

\begin{lemma}\label{lemadj}  The operators $\cdot_{r}:N_{r}\rightarrow \End (S(N))$, $\partial'_{r}: \bar{N_{r}}\rightarrow \End (S(N))$ satisfy the relation below with respect the  Hermitian metric defined in \eqref{Innersym}  
$$(\partial'_{m}v,w)_{S(N)}=( v,mw)_{S(N)},\qquad  v,w\in \Gamma(S(N)), \bar{m}\in \Gamma(\bar{N}).$$
\end{lemma}
\begin{proof} The lemma follows from the identity \vspace{-0,3cm}
$$(n_{1}\cdots  n_{k+1},mm_{1}\cdots  m_{k})_{S(N)}=\sum_{i=1}^{k+1}(n_{i},m)(n_{1}\cdots \hat{n}_{i}\cdots  n_{k+1},m_{1}\cdots \cdots  m_{k})_{S(N)}.$$\vspace{-0,1cm}
\end{proof}

In the following lemma, we consider the extension of the operators $\wedge_{r}$ and $i_{r}$ in Definition \ref{def1} to the bundle $\wedge \bar{N}$.

\begin{lemma}\label{lemadjk}  The operators $\wedge_{r}:\bar{N}_{r}\rightarrow \End (\wedge \bar{N})$, $i_{r}: {N_{r}}\rightarrow \End (\wedge \bar{N})$ satisfy the following relations with respect  the  Hermitian metric defined in \eqref{Innerwedge}
$$(i_{\bar{m}}v,w)_{\wedge \bar{N}}=( v,\wedge_{\bar{m}}w)_{\wedge \bar{N}},\qquad v,w\in \Gamma(\wedge \bar{N}), m\in \Gamma(N_{r}).$$
\end{lemma}
\begin{proof} The lemma follows from the identity \vspace{-0,2cm}
\begin{align*}
&( \bar{n}_{1}\cdots \wedge \bar{n}_{k+1},\bar{m}\wedge \bar{m}_{1}\wedge\cdots\wedge \bar{m}_{k})_{\wedge \bar{N}}\\
&=\sum_{i=1}^{k+1}(-1)^{i+1}(\bar{n}_{i},\bar{m})(\bar{n}_{1}\wedge\cdots\hat{\bar{n}}_{i}\cdots\wedge \bar{n}_{k+1},\bar{m}_{1}\wedge\cdots \hat{\bar{m}}_{i}\cdots \wedge \bar{m}_{k})_{\wedge \bar{N}}\end{align*}
\vspace{-0,1cm}\end{proof}

Finally, we have analogously as in Definition \ref{def1} operators 
\begin{equation}\label{2.14}
\begin{split}
\cdot_{r}:\bar{N}_{r}\rightarrow \End (S(\bar{N}_{r})), \qquad \bar{\partial}'_{r}: ={N_{r}}\rightarrow \End (S(\bar{N}_{r})),
\end{split}
\end{equation} 
which satisfy analogous versions of Lemma \ref{lemmun} and Lemma \ref{lemadj}.

\subsubsection{Sheaves of operators II (Harmonic oscillator)}\label{sec2.2.3}

In this section,  we modified the operators $\partial'$ in Definition \ref{def1} and $\bar{\partial}$ in \eqref{2.14}; this is going to help us to construct harmonic oscillator operators acting on the bundle $\V$. The motivation comes  more specifically from the vacuum of the harmonic oscillator, as we explain below. We refer to \cite{taubes1989s, witten1985fermion, deligne1999quantum} for a geometric description of the relation between the harmonic oscillator and the equivariant index.

The operators $\partial'$ and $\bar{\partial}'$ are analogous to the differential operator $\partial_{z}, \partial_{\bar{z}}$ acting on the space $\CC[z,\bar{z}]$; we have on $\CC[z,\bar{z}]$ a representation of the Heisenberg algebra $\{\partial_{z}, \partial_{\bar{z}}, z,\bar{z}, Id\}$  where the highest weight vector is $1\in\CC$. On the other hand, we build the Harmonic oscillator from the annihilator operators ${i}(\partial_{z}+\bar{z}/2)$, ${i}(\partial_{\bar{z}}+{z}/2)$ and creation operator ${i}(\partial_{z}-\bar{z}/2)$, ${i}(\partial_{\bar{z}}-{z}/2)$ acting on a vector space  where the highest weight vector $\ket{0}$ is an element which satisfies ${i}(\partial_{z}+\bar{z}/2)\ket{0}={i}(\partial_{\bar{z}}+{z}/2)\ket{0}=0$ (in complex analysis $\ket{0}=e^{-z\bar{z}/2}$). Algebraically, we obtain the Harmonic oscillator if the conditions ${i}(\partial_{z}+\bar{z}/2)\ket{0}={i}(\partial_{\bar{z}}+{z}/2)\ket{0}=0$ are satisfied. We observe that a simple algebraic way to obtain this condition is replacing the differentials $\partial_{z}$ by $\partial_{z}-\bar{z}/2$ and $\partial_{\bar{z}}$ by $\partial_{\bar{z}}-{z}/2$. This motivates the following definition.

\begin{definition} \label{def3.2}
$$\partial_{r}: =\partial'_{r} -\frac{r}{2}\cdot_{r}:\bar{N_{r}}\rightarrow \End (S( N_{r})),$$
$$\bar{\partial}_{r}: =\bar{\partial}'_{r} -\frac{r}{2}\cdot_{r}:N_{r}\rightarrow \End (S( N_{r})).$$
\end{definition}

As before, we extend these operators on $\V$, we have $\partial_{r}:\bar{N_{r}}\rightarrow \End (\V)$ and $\bar{\partial}_{r}:N_{r}\rightarrow \End (\V)$.  We use the following notation
\begin{align*}
&\partial_{r}(\bar{n})=\partial_{n}\in \End(\V)\,,\qquad \bar{n}\in \Gamma(\bar{N}_{r}),\\
&\partial_{r}(n)=\partial_{\bar{n}}\in \End(\V)\,,\qquad n\in \Gamma({N}_{r}).
\end{align*}

Now, let $E,E', E''$ be complex vector bundles on $M$, and $a: E\rightarrow \End(E'')$, $b:E'\rightarrow \End(E'')$ $\CC$-linear maps. We build a map $(a+b):E+E'\rightarrow \End(E\otimes E')$ as follows: Let $x\in M$ an arbitrary point  on $M$ and $a_{x}$, $b_{x}$ and $e$ vectors in $E_{x}\rightarrow \End(E'')_{x}$, $E'_{x}\rightarrow \End(E'')_{x}$ and $E_{x}$ respectively then
\[a_{x}+b_{x} :e\mapsto a(\pi_{E}e)+b(\pi_{E'}e)\in \End(E''),\]
where $\pi_{E}:E+E' \rightarrow E$ and $\pi_{E'}:E+E'\rightarrow E'$ are projections. We say that $a+b$ is the addition of $a$ and $b$.
In particular, from Corollary \ref{Koba2} we consider the addition of the bundles maps in Definitions \ref{def1} and \ref{def3.2} on $\V$, we obtain the following $\CC$-linear maps 
\begin{equation}
\label{eq3.99}
\begin{split}
&\cdot:N\rightarrow \End (\V),\qquad\quad  \partial: \bar{N}\rightarrow \End (\V), \\
&\cdot:\bar{N}\rightarrow \End (\V),\qquad\quad \bar{\partial}: {N}\rightarrow \End (\V), \\
&\wedge:\bar{N}\rightarrow \End (\V),\qquad i_{\cdot}:{N}\rightarrow \End (\V),
\end{split}
\end{equation}
where $\cdot=\cdot_{r_1}+\cdots +\cdot_{r_{k}}:N=N_{r_1}+\cdots +N_{r_{k}}\rightarrow \End(\V)$, analogously for the other operators.  

\begin{lemma} \label{lem2.3} The operators $\cdot, \bar{\partial}:N\rightarrow \End (\V)$, $\cdot, \partial:\bar{N}\rightarrow \End (\V)$ satisfy the algebra below with respect the graded commutator
$$[\partial_{n},m]=(m,n),\quad[\partial_{\bar{n}},\bar{m}]=(\bar{m},\bar{n}), \qquad m,n\in \Gamma(N), \bar{n},\bar{m}\in \Gamma({N}).$$
\end{lemma}
This lemma follows from Lemma \ref{lemmun}.

\begin{lemma} The operators $\cdot_{r}, i_{r}, \bar{\partial}_{r}:N_{r}\rightarrow \End (\V)$ and $\cdot_{r}, \partial_{r}, \wedge_{r}:\bar{N}_{r}\rightarrow \End (\V)$ satisfy the relation below with respect the  Hermitian metric defined in $\V$    
$$(\partial_{n}v,w)=( v,-\partial_{\bar{n}}w),\quad(nv,w)=(v,\bar{n}w),\quad (\bar{n}\wedge v,w)=(v,i_{\bar{n}}w),\qquad n\in \Gamma(N_{r}), v,w\in \Gamma(V).$$
\end{lemma}
\begin{proof} 
The relation $(n\wedge v,w)=(v,i_{n}w)$ follows from Lemma \ref{lemadjk}. The relation $(nv,w)=(v,\bar{n}w)$ follows from definition  of the Hermitian product \eqref{Innerss}. 
And the relation $(\partial_{n}v,w)=( v,-\partial_{\bar{n}}w)$ follows from the indentity
 \begin{align*}&(\partial_{n}(n_{1}\cdots n_{k}\bar{m}_{l+1}\cdots \bar{m}_{l+s}), m_{1}\cdots m_{l}\bar{n}_{k+1}\cdots \bar{n}_{k+t})\\
 &\qquad +(n_{1}\cdots n_{k}\bar{m}_{l+1}\cdots \bar{m}_{l+s}, \partial_{\bar{n}}(m_{1}\cdots m_{l}\bar{n}_{k+1}\cdots \bar{n}_{k+t}))\\
 &=(\partial'_{n}(n_{1}\cdots n_{k+t}),m_{1}\cdots m_{l+s})-r(n_{1}\cdots n_{k+t},nm_{1}\cdots m_{l+s})=0,
 \end{align*}
 where $n_{1}\cdots n_{k}\bar{m}_{l+1}\cdots \bar{m}_{l+s}, m_{1}\cdots m_{l}\bar{n}_{k+1}\cdots \bar{n}_{k+t}\in S(N_{r})\otimes S(\bar{N}_{r})$. The last expression is zero because we have the identity 
$$(n_{1}\cdots  n_{k+1},mm_{1}\cdots  m_{k})=\frac{1}{r}\sum_{i=1}^{k+1}(n_{i},m)(n_{1}\cdots \hat{n}_{i}\cdots n_{k+1},m_{1}\cdots . m_{k}).$$\vspace{-0,3cm}
on the Hermitian metric \eqref{Innerss}.

\end{proof}

We define the following operators

\begin{definition} \label{oscilator}
\begin{align*}
&\psi_{r}=i_{r}:{N}_{r}\rightarrow \End (\wedge\bar{ N}_{r}),\quad\quad \quad\quad \quad\quad \quad\quad\quad\quad\psi_{-r}=\wedge_{r}:\bar{N}_{r}\rightarrow \End (\wedge\bar{ N}_{r}),\\
&a_{r}:=i(\bar{\partial}_{r}+\frac{r}{2}\cdot) : {N}_{r}\rightarrow \End (S(N_{r})\otimes S(\bar{N}_{r})),\quad\quad a_{-r}:=i(\partial_{r}-\frac{r}{2}	\cdot ): \bar{N}_{r}\rightarrow \End (S(N_{r})\otimes S(\bar{N}_{r})),\\
&b_{r}:=i(\partial_{r}+\frac{r}{2}\cdot) : \bar{N_{r}}\rightarrow \End (S(N_{r})\otimes S(\bar{N}_{r})),\quad\quad b_{-r}:=i(\bar{\partial}_{r}-\frac{r}{2}\cdot) : {N_{r}}\rightarrow \End (S(N_{r})\otimes S(\bar{N}_{r})).
\end{align*}
\end{definition}

In the proposition below, we extend the operators in Definition \ref{oscilator} on $\V$.  From now on, we use the notation $N_{-r}=\bar{N}_{r}$ and we define the set $-\mathcal{A}:=\{-r|r\in \mathcal{A}\}$.

\begin{proposition} \label{ad1}The operators $a_{r}, b_{r}, \psi_{r}:{N}_{r}\rightarrow \End (\V)$ for $r\in\mathcal{A}\cup -\mathcal{A}$  satisfy the algebra below with respect the the graded commutator
$$[a_{n},a_{m}]=r(m,n)\delta_{r,-l}, \quad [b_{n},b_{m}]=r(m,n)\delta_{r,-l}, \quad [\psi_{n},\psi_{m}]=(m,n)\delta_{r,-l},\qquad n\in \Gamma(N_{r}), m\in \Gamma(N_{l}),$$
where  $r,l\in \mathcal{A}\cup -\mathcal{A} $.
\end{proposition}

\begin{proposition}\label{ad} The operators $a_{r}, b_{r}, \psi_{r}:{N}_{r}\rightarrow \End (\V)$ for $r\in\mathcal{A}\cup -\mathcal{A}$ satisfy the relation below with respect the  Hermitian metric defined in $\V$
$$(a_{n} v,w)=(v,a_{\bar{n}}v),\quad (b_{n} v,w)=(v,b_{\bar{n}}v), \quad (\psi_{n}v,w)=(v,\psi_{\bar{n}}w),\qquad n\in \Gamma(N_{r}), v,w\in \Gamma(V).$$
\end{proposition}

Finally,  we introduce the following notation. The Clifford bundle $CL(TM)$ on $M$ acts on the bundle $\bigtriangleup(M)$. We denote by $\psi'_{0}:TM\rightarrow CL(TM)\rightarrow \End (\bigtriangleup(M))$ the map composition, this maps unit tangent vectors to unitary operators on $\bigtriangleup(M)$. We denote the extension of the Clifford action from $\bigtriangleup (M)$ to $\V$ as follows 
\begin{equation}\psi_{0}:TM\rightarrow \End (\V),\label{clifzero}\end{equation} 
and by definition of the Hermitian product on $V$ we have 
$$(\psi_{0}(x)v,w)=-(v,\psi_{0}(x)w),\qquad x\in \Gamma(TM_{\mathbb{C}}),\quad   v,w\in \Gamma(V).$$

\subsection{Global operators}\label{2.2.4} 
In this section, we define the global operators, we are motivated by the global operators \eqref{eq2.16} and \eqref{eq2.17}. In order to make this section clear, we write after each definition the operator in a local basis on an open set $U\subset M$. Let $\{z^{k}_{r}\}_{k=1,r\in\mathcal{A}}^{d_{r}}$ an orthonormal basis of $N=\oplus_{r\in\mathcal{A}}N_{r}$ on $U$, and its complex conjugates $\{\bar{z}^{k}_{r}\}_{k=1,r\in-\mathcal{A}}^{d_{r}}$ a basis of $\bar{N}=\oplus_{r\in\mathcal{A}}\bar{N}_{r}=\oplus_{r\in\mathcal{A}}{N}_{-r}$ on $U$. Then 
\begin{equation}\label{3.12}
\begin{split}
&\cdot({z^{k}_{r}})=z^{k}_{r}\in \End (\V)|_{U},\quad\quad \cdot({\bar{z}^{k}_{r}})=\bar{z}^{k}_{r}\in \End (\V)|_{U},\\
&\partial(z^{k}_{r})=\partial_{\bar{z}^{k}_{r}}\in \End (\V)|_{U},\quad\quad \partial(\bar{z}^{k}_{r})=\partial_{{z}^{k}_{r}}\in \End (\V)|_{U},\\
& i_{.}({z}^{k}_{r})=\psi^{k}_{r} \in\End (\V)|_{U},\quad\quad \wedge_{r}(\bar{z}^{k}_{r})=\psi^{k}_{-r}\in  \End (\V)|_{U}. 
\end{split}
\end{equation}
Additionally, we consider an orthonormal basis $\{e_{i}\}_{i=1}^{2l}$ of $TM$ on $U$, and $\psi_{0}(e_{i})=\psi^{i}_{0}\in \End (\V)|_{U}$. 

\begin{definition}\label{def3.13}
The twisted Dirac operator on $\slashed{D}_{M}:\Gamma(\V)\rightarrow \Gamma(\V)$ is defined by the composition
\begin{equation}\slashed{D}_{M}:\Gamma(\V)\xrightarrow{D^{M}}\Gamma (\V\otimes T^{*}M) \rightarrow \Gamma (\V\otimes TM)\xrightarrow{\psi_{0}} \Gamma ( \V)\label{g1}\end{equation}
 where we used the connection $D^{M}$ in $\eqref{connection}$, $\eqref{clifzero}$ and the Riemannian metric on $M$ to identified $TM\simeq T^{*}M$.
 \end{definition}
 Locally this operator is written on $U$ as follows 
 $$\slashed{D}_{M}|_{U}=\psi^{i}_{0}D^{M}_{e_{i}}.$$
 
 Note that $\Hom(\bar{N},\Hom(\V,\V))=\Hom(\V,\V\otimes \bar{N}^{*})$, hence the $\CC$-linear maps in Definition \ref{def3.2} are equivalently expressed as $\CC$-linear maps $\partial: \V\rightarrow \V\otimes \bar{N}^{*} $ and $\bar{\partial}: \V\rightarrow \V\otimes {N}^{*} $. This description is useful to generalize the connection $D^{M}$: The connection $D^{M}$ in \eqref{connection} is defined for vectors in the tangent bundle $TM$, we consider more generally derivatives in the normal direction (see \eqref{eq3.99}) as follows
$$ D^{X}:=D^{M}+i\partial+i\bar{\partial}:\Gamma(\V)\rightarrow \Gamma(\V\otimes (T^{*}M+ N^{*}+ \bar{N}^{*})).$$
 Locally this operator is written using the basis \eqref{3.12} on $U$ as follows
$$ D^{X}|_{
U}=D_{e_{i}}^{M}+\sum_{r\in\mathcal{A}}\sum_{k=1}^{d_{r}}i\partial_{\bar{z}^{k}_{r}}+i\partial_{z^{k}_{r}}.$$

\begin{definition}\label{D}The Dirac operator on $\slashed{D}_{X}:\Gamma(\V)\rightarrow \Gamma(\V)$ is defined by the composition
\begin{equation}\slashed{D}_{X}:\Gamma(\V)\xrightarrow{D^{X}}\Gamma(\V\otimes (T^{*}M\bigoplus_{r\in\mathcal{A}}{N}_{r}^{*}+ \bar{N}_{r}^{*}) )\rightarrow \Gamma(\V\otimes (TM\bigoplus_{r\in\mathcal{A}}\bar{N}_{r}+ {N}_{r}))\xrightarrow{\psi_{0}+\psi_{r}} \Gamma(\V)\label{g22}\end{equation}
 where $TM\simeq T^{*}M$ we use the Riemannian metric on $M$ to identified $TM\simeq T^{*}M$ and the Hermitian metric \eqref{g} to identify $N^{*}_{r}\simeq \bar{N}_{r}$ and $\bar{N}^{*}_{r}\simeq N_{r}$. For abuse of language we denote this Dirac operator as $\slashed{D}_{X}$. 
 \end{definition}
 Locally this operator is written using the basis \eqref{3.12} on $U$ as follows
$$ \slashed{D}_{X}|_{
U}=\psi^{i}_{0}D^{M}_{e_{i}}+\sum_{r\in\mathcal{A}}\sum_{k=1}^{d_{r}}\psi^{k}_{-r}i\partial_{z^{k}_{-r}}+\psi^{k}_{r}i\partial_{z^{k}_{r}}.$$
 
 Let $E$ be a complex vector bundle on $M$, $\{U_{\alpha}\}_{\alpha\in I}$ a covering of $M$ and $g_{\alpha,\beta}:U_{\alpha}\times U_{\beta}\rightarrow U(E)$ transition functions of $E$ for $\alpha,\beta\in I$. If  $a_{\alpha}:E|_{U_{\alpha}}\rightarrow \End(E)|_{U_{\alpha}}$ are $\CC$ linear maps such that  for arbitrary $e, e'\in E$ we have $g_{\alpha,\beta}(a_{\alpha}(e)(e'))=a_{\beta}(g_{\alpha,\beta} e)(g_{\alpha,\beta}e')$ then it exists a global $\CC$ linear map $a:E \rightarrow \End(E)$ such that $a|_{U_{\alpha}}=a_{\alpha}$ for all $\alpha\in I$.
 
 In particular, we are motivated by $K$, $K'$ and $dK'$ to define the algebraic $\CC$-linear maps $K\in\End(\V)$, $K'\in\End(\V)$ and $dK'\in\End(\V)$ as follows
$$ K|_{U}=\sum_{r\in\mathcal{A}}\sum_{k=1}^{d_{r}}r(z^{k}_{r}\partial_{z^{k}_{r}}-{z^{k}}_{-r}\partial_{z^{k}_{-r}}),\quad\quad K'|_{U}=\sum_{r\in\mathcal{A}}\sum_{k=1}^{d_{r}}r(z^{k}_{r}\psi^{k}_{-r}-{z^{k}}_{-r}\psi^{k}_{r}),$$
$$ dK'|_{U}=\sum_{r\in\mathcal{A}}\sum_{k=1}^{d_{r}}r(\psi^{k}_{r}\psi^{k}_{-r}-{\psi^{k}}_{-r}\psi^{k}_{r})=c_{1}-2\sum_{r\in\mathcal{A}}\sum_{k=1}^{d_{r}}r{\psi^{k}}_{-r}\psi^{k}_{r}.$$
Note that the value $c_{1}=\sum_{r\in\mathcal{A}}\sum_{k=1}^{d_{r}}r$ appeared in the expression above (see \eqref{ASS}).

Moreover, we express the operators in Definition \ref{oscilator} locally on $U$ as follows
$$\alpha^{k}_{r}:=i(\partial_{z^{k}_{-r}}+\frac{r}{2}z^{k}_{r})\in \End (\V)|_{U},\quad\quad \alpha^{k}_{-r}:=i(\partial_{z^{k}_{r}}-\frac{r}{2}z^{k}_{-r})\in \End (\V)|_{U}, $$
$$\beta^{k}_{r}:=i(\partial_{z^{k}_{r}}+\frac{r}{2}z^{k}_{-r})\in \End (\V)|_{U},\quad\quad \beta^{k}_{-r}:=i(\partial_{z^{k}_{-r}}-\frac{r}{2}z^{k}_{r})\in  \End (\V)|_{U}. $$

Now, we define the most important operators in this work. 
\begin{definition} \label{global1}
\begin{equation}{{Q:=\slashed{D}_{X}+\frac{i}{2}K'}\in \End (\Gamma(\V)),\quad\quad {L_{K}:=D^{X}_{K}+\frac{1}{2}dK'}}\in \End (\V).\end{equation}
\end{definition}
 Locally the operator Q is written on $U$ as follows
\begin{equation} \label{global1}
\begin{split}
Q|_{U}&=\psi^{i}_{0}D_{e_{i}}+\sum_{r\in\mathcal{A}}\sum_{k=1}^{dimN_{r}}\psi^{k}_{-r}i(\partial_{z^{k}_{-r}}+\frac{r}{2}z^{k}_{r})+\psi^{k}_{r}i(\partial_{z^{k}_{r}}-\frac{r}{2}z^{k}_{-r})\\
&=\psi^{i}_{0}D_{e_{i}}+\sum_{r\in\mathcal{A}}\sum_{k=1}^{dimN_{r}}\psi^{k}_{r}\alpha^{k}_{-r}+\psi^{k}_{-r}\alpha^{k}_{r},
\end{split}
\end{equation}
analogously,  the operator $L_{K}$ is written on $U$ as follows
\begin{align*}L_{K}|_{U}&=\frac{1}{2}c_{1}+\sum_{r\in\mathcal{A}}\sum_{k=1}^{d_{r}}\beta^{k}_{-r}\beta^{k}_{r}-\alpha^{k}_{-r}\alpha^{k}_{r}-r\psi^{k}_{-r}\psi^{k}_{r}.
\end{align*}
We finish this subsection introducing the notation $L_{\alpha}:=\sum_{r\in\mathcal{A}}\sum_{k=1}^{d_{r}}\alpha^{k}_{-r}\alpha^{k}_{r}$, $L_{\beta}:=\sum_{r\in\mathcal{A}}\sum_{k=1}^{d_{r}}\beta^{k}_{-r}\beta^{k}_{r}$ and  $L_{\psi}=\sum_{r\in\mathcal{A}}\sum_{k=1}^{d_{r}}r\psi^{k}_{-r}\psi^{k}_{r}$, hence
\begin{equation}\label{c1}L_{K}=\frac{1}{2}c_{1}+L_{\beta}-L_{\alpha}-L_{\psi}.
\end{equation}

\subsection{Index and supersymmetry}\label{sec3.4}

We define an Hermitian product on the vector space of global sections $\Gamma(\V)$ as follows  
$$\langle v, w\rangle :=\int_{M}(v(x),w(x))dM,\qquad  v,w\in \Gamma(\V).$$
Where $dM$ is the volume form coming form the Riemannian metric on M. From Lemma \ref{ad}, the operator defined by $\slashed{D}_{N}:=Q-\slashed{D}_{M}$ is self adjoint, moreover $\slashed{D}_{M}$ is a twisted Dirac operator and its self adjoinest follows from the divergence for vector fields. We obtain the following proposition

\begin{proposition}\label{i} The operator $Q$ is self adjoint i.e.
$$\langle Qv,w\rangle =\langle v,Qw\rangle ,\qquad v,w\in \Gamma(\V).$$
\end{proposition}

We are motivated by  \eqref{equiin} to define \emph{the analytic equivariant index} of $Q$ as follows. 
$$\Ind (Q,q):=\Tr (q^{L_{K}}|_{\ker (Q|_{\V^{+}})})-\Tr (q^{L_{K}}|_{\ker (Q|_{\V^{-}})}), $$
 where $q=e^{i\theta}\in\mathbb{C}^{\times}$. 

Now, the following lemmas will be useful to prove that the index above is equal to the right hand side of \eqref{ASS}. 

Note that the operator $q^{L_{K}}$ is diagonal on $\V$, in particular we have 
$$q^{L_{K}}S(N)=q^{\frac{1}{2}c_{1}}q^{L_{\beta}}S(N)=q^{\frac{1}{2}c_{1}}\otimes_{r\in\mathcal{A}}S_{q^{r}}(N_{r}),$$
$$q^{L_{K}}S(\bar{N})=q^{\frac{1}{2}c_{1}}q^{-L_{\alpha}}S(\bar{N})=q^{\frac{1}{2}c_{1}}\otimes_{r\in\mathcal{A}}S_{q^{-r}}(\bar{N}_{r}), $$ 
$$q^{L_{K}}\wedge^{*}(\bar{N})=q^{\frac{1}{2}c_{1}}q^{-L_{\psi}}\wedge^{*}(\bar{N})=q^{\frac{1}{2}c_{1}}(\otimes_{r\in\mathcal{A}}\wedge^{\even}_{q^{-r}}(\bar{N}_{r})-\otimes_{r\in\mathcal{A}}\wedge^{\odd}_{q^{-r}}(\bar{N}_{r})). $$ 

\begin{lemma} \label{operalg}The operators $Q$ and $L_{K}$ satisfy the following relations

\emph{(i)}\;  $[Q,L_{K}]=0$,

\emph{(ii)}\;  $Q^{2}=\slashed{D}^{2}_{M}+L_{\alpha}+L_{\psi}.$
\end{lemma}

\begin{proof} It is enough to proof  (i) and (ii) on an open set on $M$. As before, let $\{z^{k}_{r}\}_{k=1,r\in\mathcal{A}}^{d_{r}}$ be an orthonormal basis of $N=\oplus_{r\in\mathcal{A}}N_{r}$ on $U$, and its complex conjugates $\{\bar{z}^{k}_{r}\}_{k=1,r\in-\mathcal{A}}^{d_{r}}$ a basis of $\bar{N}=\oplus_{r\in\mathcal{A}}\bar{N}_{r}=\oplus_{r\in\mathcal{A}}{N}_{-r}$ on $U$, and $\{e_{i}\}_{i=1}^{2l}$ an orthonormal basis of $TM$ on $U$

First, we prove (ii). Let $\nabla^{X}$ be the connection on $N_{r}$ and $\bar{N}_{r}$ defined in \eqref{con}, then on the local basis we have  
 $$\nabla^{X}_{e_{i}}z^{j}_{r}=w_{ij}^{k}(r)z^{k}_{r}\,,\qquad \nabla^{X}_{e_{i}}\bar{z}^{j}_{r}=w_{ij}^{k}(r)\bar{z}^{k}_{r},$$
 where $w_{ij}^{k}(r)\in \RR$. Note that these connections were used to define the connections in \eqref{con2} and $D^{M}:\Gamma(\V)\rightarrow \Gamma(\V)$ in \eqref{connection}. Then, using the notation in \eqref{3.12}, we have that the commutator satisfies the following algebra $[D^{M}_{e_{i}},z^{j}_{r}]=w_{ij}^{k}(r)z^{k}_{r}$, $[D^{M}_{e_{i}},\psi^{j}_{r}]=-w_{ik}^{j}(r)\psi^{k}_{r}$,  $[D^{M}_{e_{i}},\partial_{z^{j}_{r}}]=-w_{ik}^{j}(r)\partial_{z^{k}_{r}}$ and $[D^{M}_{e_{i}},\psi^{j}_{-r}]=w_{ij}^{k}(r)\psi^{k}_{-r}$. And considering their conjugates we have $[D^{M}_{e_{i}},\bar{z}^{j}_{r}]=w_{ij}^{k}(r)\bar{z}^{k}_{r}$ and $[D^{M}_{e_{i}},\partial_{\bar{z}^{j}_{r}}]=-w_{ik}^{j}(r)\partial_{\bar{z}^{k}_{r}}$. Therefore, we have the following relations for the operators in Definition \ref{oscilator}
$$[D^{M}_{e_{i}},\psi^{j}_{r}]=-w_{ik}^{j}(r)\psi^{k}_{r},\quad\quad [D^{M}_{e_{i}},\psi^{j}_{-r}]=w_{ij}^{k}(r)\psi^{k}_{-r},$$
$$[D^{M}_{e_{i}},\alpha^{j}_{r}]=-w_{ik}^{j}(r)\alpha^{k}_{r},\quad\quad [D^{M}_{e_{i}},\alpha^{j}_{-r}]=w_{ij}^{k}(r)\alpha^{k}_{-r},$$
$$[D^{M}_{e_{i}},\beta^{j}_{r}]=-{w}_{ik}^{j}(r)\beta^{k}_{r},\quad\quad [D^{M}_{e_{i}},\beta^{j}_{-r}]={w}_{ij}^{k}(r)\beta^{k}_{-r}.$$
This implies the following relations 
$$[\slashed{D}_{M},\slashed{D}_{N}]=\psi^{i}_{0}[D^{M}_{e_{i}},\psi^{j}_{r}]\alpha^{j}_{-r}-\psi^{j}_{r}\psi^{i}_{0}[D^{M}_{e_{i}},\alpha^{j}_{-r}]=\psi^{i}_{0}(-w^{j}_{ik}(r)\psi^{k}_{r})\alpha^{j}_{-r}+\psi^{i}_{0}\psi^{j}_{r}(w_{ij}^{k}(r)\alpha^{k}_{-r})=0,$$
and
$$\slashed{D}_{N}^{2}=[\slashed{D}_{N},\slashed{D}_{N}]=2\sum_{r\in\mathcal{A}}\sum_{k=1}^{\dim N_{r}}\alpha^{k}_{-r}\alpha^{k}_{r}+r\psi^{k}_{-r}\psi^{k}_{r}+r-r=2L_{\alpha}+2L_{\psi}.$$
This proves (ii).
 
Now we prove (i), we have the following relations
$$[\slashed{D}_{N},L_{\alpha}+L_{\psi}]=[\psi_{r}\alpha_{-r},\alpha_{n}\alpha_{-n}+n\psi_{n}\psi_{-n}]=-r\psi_{r}\alpha_{-r}+r\psi_{r}\alpha_{-r}=0,\quad\quad [\slashed{D}_{N},L_{\beta}]=0,$$ 
$$[\slashed{D}_{M},L_{\alpha}]=\psi^{i}_{0}(-w_{ik}^{j}(n)\alpha^{k}_{n})\alpha^{j}_{-n}+\psi^{i}_{0}\alpha^{j}_{n}(w^{k}_{ij}(n)\alpha^{k}_{-n})=0,$$
$$[\slashed{D}_{M},L_{\psi}]=\psi^{i}_{0}(-w_{ik}^{j}(n)\psi^{k}_{n})\psi^{j}_{-n}+\psi^{i}_{0}\psi^{j}_{n}(w^{k}_{ij}(n)\psi^{k}_{-n})=0,$$
$$[\slashed{D}_{M},L_{\beta}]=\psi^{i}_{0}(-{w}_{ik}^{j}(n)\beta^{k}_{n})\beta^{j}_{-n}+\psi^{i}_{0}\beta^{j}_{n}({w}^{k}_{ij}(n)\beta^{k}_{-n})=0,$$
hence $[Q,L_{K}]=0$.
\end{proof}

Additionally, we have the following lemma  

\begin{lemma}[susy]\label{susy}
{$\ker (Q)\subset \V':=\bigtriangleup (M)\otimes \sqrt{\det N}\otimes S(N)\subset \V$}
\end{lemma}

\begin{proof}
As we mentioned above $L_{\alpha}$ and $L_{\psi}$ are diagonal endomorphisms in $\V$, in particular the eigenvalues of $L_{\alpha}$ and $L_{\psi}$ are strictly positive on $\wedge^{k}\bar{N}\otimes S(\bar{N})$. 
Hence, for any element $v\in \Gamma(V)$ we have that
$$\langle Qv,Qv\rangle =\langle Q^{2}v,v\rangle =\langle \slashed{D}_{M}v,\slashed{D}_{M}v\rangle +\langle L_{\alpha}+L_{\psi}v,v\rangle \geq \langle L_{\alpha}+L_{\psi}v,v \rangle$$
where we used Proposition \ref{i} and Lemma \ref{operalg}. 
\end{proof}

Finally, we have the theorem that express the index in terms of a topological expression
    
    \begin{theorem}\label{index}
    $$ {\Ind (Q,q)=q^{\frac{c_{1}}{2}}\int_{M} A(M)\ch ( \sqrt{\det N}\otimes_{r\in\mathcal{A}} S_{q^{r}}(N_{r}))}$$
\end{theorem}
\begin{proof}
From Theorem \ref{susy} we have that $\ker (Q)\subset \V'$ and by definition of $Q$ (see \eqref{global1}) and by definition of $\slashed{D}_{M}$ (Definition \ref{def3.13}) we have that $Q|_{\V'}=\slashed{D}_{M}|_{\V'}$, therefore 
$$\Ind (Q,q)=\Ind (\slashed{D}_{M}|_{\V'},q)=\Tr (q^{L_{K}}|_{\ker (\slashed{D}_{M}|_{\V'_{+}})})-\Tr (q^{L_{K}}|_{\ker (\slashed{D}_{M}|_{\V'_{-}} ) }).$$
Now, we have that 
$$q^{L_{K}}(\bigtriangleup (M)\otimes \sqrt{\det N}\otimes S(N))=q^{\frac{c_{1}}{2}}\bigtriangleup (M)\otimes \sqrt{\det N}\otimes_{r\in\mathcal{A}} S_{q^{r}}(N_{r}).$$
And from Theorem $\ref{operalg}$ (i), the eigenvalues of $L_{K}$ commute with $\slashed{D}_{M}$ then 
$$ \Ind (\slashed{D}_{M}|_{\V'},q)=q^{\frac{c_{1}}{2}}\Ind (\slashed{D}_{M}\otimes \sqrt{\det N}\otimes_{r\in\mathcal{A}} S_{q^{r}}(N_{r})).$$
Finally, the theorem follows from the Atiyah-Singer index theorem for twisted bundles \eqref{TAS}.
\end{proof}

\section{Localization, Infinite dimensional case}\label{sec4}
In this section we consider the infinite dimensional case,  $X=\L M$. We assume that $M$ is a  $2l$-dimensional, connected, Riemannian, spin manifold. 
\subsection{The Vector Bundle $\V_{R}$}\label{sec4.1}
Let $X=\L M$ be the loop space of $M$, on the loop space there is a natural $S^{1}$-action such that the set of fixed points is given by $M$. We work on the normal bundle of the manifold $M$ inside $X$, this bundle is defined as follows  
$$\mathcal{N}:= \bigoplus_{n>0} TM_{\mathbb{C}}=\bigoplus_{n>0}N_{n}, \quad\quad N_{n}=TM_{\mathbb{C}}.$$
Hence we are motivated by the finite dimensional case to define the vector space 
$$\V_{R}:=\bigtriangleup (M)\otimes_{n\in \mathbb{Z}_{+}}\wedge^{*}\bar{N}_{n}\otimes_{n\in \mathbb{Z}_{+}}S(N_{n})\otimes _{n\in \mathbb{Z}_{+}}S(\bar{N}_{n})\rightarrow M .$$
Unlike the finite dimensional case each one of the bundles $N_{n}$ are trivially complex, this in particular implies that $\det N_{n}$ is the trivial bundle for all $n>0$. Additionally, in this case we have that $N_{n}=\bar{N}_{n}= TM_{\mathbb{C}}$.  Now,  $\V_{R}=\V_{R}^{+}+\V_{R}^{-}$ is a $\mathbb{Z}_{2}$ graded bundle where
\begin{equation*}
\begin{split}
&\V^{+}_{R}=\left(\bigtriangleup^{+}(M)\otimes^{\even}_{n\in\mathbb{Z}_{+}}\wedge^{*}\bar{N}_{n}+ \bigtriangleup^{-}(M)\otimes^{\odd}_{n\in\mathbb{Z}_{+}}\wedge^{*}\bar{N}_{n} \right)\otimes_{n\in\mathbb{Z}_{+}} S(N_{n})\otimes_{n\in\mathbb{Z}_{+}} S(\bar{N}_{n}),\\
&\V^{-}_{R}=\left(\bigtriangleup^{-}(M)\otimes^{\even}_{n\in\mathbb{Z}_{+}}\wedge^{*}\bar{N}_{n}+ \bigtriangleup^{+}(M)\otimes^{\odd}_{n\in\mathbb{Z}_{+}}\wedge^{*}\bar{N}_{n}\right)\otimes_{n\in\mathbb{Z}_{+}} S(N_{n})\otimes_{n\in\mathbb{Z}_{+}} S(\bar{N}_{n})
\end{split}
\end{equation*}

\subsubsection{Hermitian product}\label{IP2}

The vector bundle $\V_{R}$ is defined as a tensor product in terms of other bundles, we define an Hermitian product on $\V_{R}$ considering Hermitian products for each one of the bundles which define $\V_{R}$, we proceed as in Section \ref{IP}.

First, as we already mentioned in the previous section, $M$ is a spin manifold then its spin bundle $\bigtriangleup (M)$ has an Hermitian product. 

Second, on $N_{n}=TM_{\mathbb{C}}$ for arbitrary $n>0$ one has a natural Hermitian metric defined as follows
\begin{equation}
(av,bu)=a\bar{b}g_{M}(u,v),\label{Rinner}
\end{equation}
where $a,b\in\mathbb{C}$, $u,v\in \Gamma(TM)$ and $g_{M}$ is the Riemannian metric on $M$. The Hermitian metric on $N_{n}$ is defined by $({n},{m})_{N_{n}}:=(n,m)$ where $n,m\in\Gamma(N_{n})$, and the Hermitian metric on $\bar{N}_{n}$ is given by $(\bar{n},\bar{m})_{\bar{N}_{n}}:=({m},{n})$ where $\bar{n},\bar{m}\in\Gamma(\bar{N}_{n})$. The Hermitian structure on $ \wedge^{*}N_{n}$ is defined by \eqref{Innerwedge} and the Hermitian structure on $S(N_{n})\otimes S(\bar{N}_{n})$ is defined by \eqref{Innerss}.

Finally, the Hermitian product on $\V_{R}$ is given by the tensor product of the inner products on $\bigtriangleup(M)$, $\wedge^{*}\bar{N}_{n}$ and $S(N_{n})\otimes S(\bar{N}_{n})$ for $n>0$. Therefore $\V_{R}$ is an Hermitian bundle, we denote the Hermitian product by $(\cdot,\cdot)$.

\subsection{Operators }\label{sec4.2}

In this section, we define the operators acting on the bundle $\V_{R}$.

\subsubsection{Connections} Now, we build an Hermitian connection on $\V_{R}$, once again we use that $\V_{R}$ is defined as a tensor product in terms of other bundles, we proceed as in Section \ref{sec3.2.1}.

First, like in the finite dimensional case on the bundle $\Gamma(\bigtriangleup(M))$ we have a connection $D^{\bigtriangleup}:\Gamma(\bigtriangleup(M))\rightarrow\Gamma(\bigtriangleup(M)\otimes T^{*}M)$, this connection is compatible with the Hermitian metric $\bigtriangleup(M)$.

Second, we obtain from the Levi-Civita connection $\nabla^{M}$ on $M$ the Hermitian connections
\begin{equation}\label{coninfi}
\begin{split}
&\nabla^{M}:\Gamma(N_{n})\rightarrow \Gamma(N_{n}\otimes   T^{*}M),\qquad \nabla^{M}:\Gamma(\bar{N}_{n})\rightarrow \Gamma(\bar{N}_{n}\otimes   T^{*}M),
\end{split}
\end{equation}
for $n>0$.

Third, we obtain analogously compatible connections like in \eqref{con2} for $n>0$, i.e. $D^{S}:  \Gamma(S(N_{n}))\rightarrow \Gamma(S( N_{n})\otimes   T^{*}M )$, 
$D^{\bar{S}}: \Gamma(S( \bar{N}_{n}))\rightarrow \Gamma(S(\bar{N}_{n})\otimes   T^{*}M)$ and
$D^{\wedge^{*}}:\Gamma(\wedge^{*}(\bar{N}_{n}))\rightarrow \Gamma(\wedge^{*}(\bar{N}_{n})\otimes   T^{*}M)$ with respect their Hermitian metrics.

Finally, from the tensor product of these connections, we obtain a connection on $\V_{R}$
\begin{equation}\label{dm2}
D^{M}:\Gamma(\V_{R})\rightarrow \Gamma(\V_{R}\otimes T^{*}M)
\end{equation} 
And this connection is compatible with the Hermitian metric $(\cdot,\cdot)$ on $\V_{R}$.
\subsubsection{Sheaves of operators}
All the definitions and lemmas in Section \ref{sheope} and section \ref{sec2.2.3}  works on $\V_{R}$. We restate the definition \ref{oscilator} considering that the maps are extended to $\V_{R}$, in this case the maps are defined on $N_{n}$ for arbitrary $n>0$.

\begin{definition} \label{oscilator2}
\begin{align*}
&\psi_{n}=i_{n}:{N}_{n}\rightarrow \End (\V_{R}),\quad\quad \quad\quad \quad\quad \quad\quad\quad\psi_{-n}=\wedge_{n}:\bar{N}_{n}\rightarrow \End (\V_{R}),\\
&a_{n}:=i(\bar{\partial}_{n}+\frac{n}{2}\cdot) : {N}_{n}\rightarrow \End (\V_{R}),\quad\quad a_{-n}:=i(\partial_{n}-\frac{n}{2}	\cdot ): \bar{N}_{n}\rightarrow \End (\V_{R}),\\
&b_{n}:=i(\partial_{n}+\frac{n}{2}\cdot) : \bar{N_{n}}\rightarrow \End (\V_{R}),\quad\quad b_{-n}:=i(\bar{\partial}_{n}-\frac{n}{2}\cdot) : {N_{n}}\rightarrow \End (\V_{R}).
\end{align*}
\end{definition}
Propositions \ref{ad1} and \ref{ad} are analogously satisfied. The Clifford algebra on $CL(TM)$ is extended to $\V_{R}$ as follows
\begin{equation}\label{cli2}
\psi_{0}:TM\rightarrow \End (\V_{R}).
\end{equation}
And we have that
$$(\psi_{0}(x)v,w)=-(v,\psi_{0}(x)w),\qquad x\in\Gamma(TM_{\mathbb{C}}),\quad v,w\in \Gamma(\V_{R}).$$

\subsection{Global operators}\label{sec4.3}

Now, we define the global operators. In order to make this section clear, we write after each definition the operator in a local basis on an open set $U\subset M$. Let  $\{z^{k}_{n}\}_{k=1}^{2l}$ an orthonormal basis of $\N=\oplus_{n>0}N_{n}$ on $U$, and its complex conjugates $\{\bar{z}^{k}_{n}\}_{k=1}^{2l}$ a basis of $\bar{\N}=\oplus_{n>0}\bar{N}_{n}$ on $U$. Then
\begin{equation}\label{alpha1}
\begin{split}
&\alpha^{k}_{n}:=i(\partial_{z^{k}_{-n}}+\frac{r}{2}z^{k}_{n})\in\End (\V_{R})|_{U},\quad\quad \alpha^{k}_{-n}:=i(\partial_{z^{k}_{n}}-\frac{r}{2}z^{k}_{-n})\in\End (\V_{R})|_{U}, \\
&\beta^{k}_{n}:=i(\partial_{z^{k}_{n}}+\frac{r}{2}z^{k}_{-n})\in\End (\V_{R})|_{U},\quad\quad \beta^{k}_{-n}:=i(\partial_{z^{k}_{-n}}-\frac{r}{2}z^{k}_{n})\in  \End (\V_{R})|_{U}. \\
&\psi^{k}_{n}\in i_{n}\in \End (\V_{R})|_{U},\quad\quad\quad\qquad\qquad \psi^{k}_{-n}\in \wedge_{n}\in \End (\V_{R})|_{U}. 
\end{split}
\end{equation}
We have the following version of Definition \ref{def3.13}.
\begin{definition}
The twisted Dirac operator on $\slashed{D}_{M}:\Gamma(\V_{R})\rightarrow \Gamma(\V_{R})$ is defined by the composition
\begin{equation}\slashed{D}_{M}:\Gamma(\V_{R})\xrightarrow{D^{M}}\Gamma (\V_{R}\otimes T^{*}M) \rightarrow \Gamma (\V_{R}\otimes TM)\xrightarrow{\psi_{0}} \Gamma ( \V_{R})\label{g1}\end{equation}
 where we used the connection $D^{M}$ in $\eqref{dm2}$, $\eqref{cli2}$ and the Riemannian metric on $M$ to identified $TM\simeq T^{*}M$.
 \end{definition}
 Locally this operator is written on $U$ as follows 
 $$\slashed{D}_{M}|_{U}=\psi^{i}_{0}D^{M}_{e_{i}}.$$

Because in this case we work with the infinite set $\mathbb{Z}_{+}$ instead of the finite set $\mathcal{A}$, some operators analogous to the operators in Section \ref{2.2.4} are not well defined. In order to solve this problem, we define the \emph{normal ordered product} as follows (In particular Definition \ref{4.1} is not well defined without the normal ordered product). Note that the sections of the vector bundle $\V_{R}$ on $U$ satisfies that 
$$(\V_{R})_{U}\cong \mathbb{C}[\psi^{k_{1}}_{-1},\psi^{k_{2}}_{-2},\cdots .]\otimes \mathbb{C}[\alpha^{j_{1}}_{-1},\alpha^{j_{2}}_{-2},\cdots .]\otimes \mathbb{C}[\beta^{i_{1}}_{-1},\beta^{i_{2}}_{-2},\cdots .]\otimes C^{\infty}(U),$$
where $k_{1},k_{2},\cdots .j_{1},j_{2},\cdots .i_{1}, i_{2}\in\{1,\cdots ,2l\}$. Using this local form of the vector space, the normal order product is defined as a composition of operators writing the {annihilation operators} $\alpha^{k}_{n}, \beta^{k}_{n}$  and $\psi^{k}_{n}$ $(n>0)$ to the right of the {creation operators}  $\alpha^{k}_{-n}, \beta^{k}_{-n}$  and $\psi^{k}_{-n}$. We use the symbol $:(-):$ to indicate the normal ordering. In particular, for $n >0$ and arbitrary $k,j \in\{1,\cdots ,2l\}$ we have that 
$$ :\alpha^{k}_{-n}\alpha^{j}_{n}: = :\alpha^{j}_{n}\alpha^{k}_{-n}: =\alpha^{k}_{-n}\alpha^{j}_{n},\quad   :\beta^{k}_{-n}\beta^{j}_{n}: = :\beta^{j}_{n}\beta^{k}_{-n}: = \beta^{k}_{-n}\beta^{j}_{n},\quad  :\psi^{k}_{-n}\psi^{j}_{n}: = :\psi^{j}_{n}\psi^{k}_{-n}: =\psi^{k}_{-n}\psi^{j}_{n}. $$

Note that the local definition above of the normal ordered product is well defined globally for $a_{n}, b_{-n}, \phi_{-n}: N_{n}\rightarrow \End (\V_{R})$ and $b_{n}, a_{-n}, \psi_{n}:: \bar{N}_{n}\rightarrow \End (\V_{R})$ for $n>0$.

The following definition is motivated by \eqref{global1}.

\begin{definition}\label{4.1}   The Ramond Dirac operator $Q_{R}:\Gamma(\V_{R})\rightarrow \Gamma(\V_{R})$ is defined as the normal ordered product of the composition 
\begin{equation*}
\Gamma(\V_{R})\xrightarrow{D^{M}+a_{n}}\Gamma(\V_{R}\otimes (T^{*}M\oplus_{n\in\mathbb{Z}_{+}} N_{n}^{*}\oplus  \bar{N}_{n}^{*}) )\rightarrow \Gamma(\V_{R}\otimes (TM\oplus_{n\in\mathbb{Z}_{+}} \bar{N}_{n}\oplus{N}_{n}))\xrightarrow{\psi_{0}+\psi_{n}} \Gamma(\V_{R}),
\end{equation*}
 where we use the Riemannian metric on $M$ to identified $TM\simeq T^{*}M$ and the Hermitian metric \eqref{Rinner} to identify $N^{*}_{r}\simeq \bar{N}_{r}$ and $\bar{N}^{*}_{r}\simeq N_{r}$.
\end{definition}
 Locally the operators $Q_{R}$ is written on $U$ as follows
  $${ Q_{R}|_{U}=\psi^{i}_{0}\nabla_{e_{i}}+\sum_{n\in\mathbb{Z}_{+}}\psi^{k}_{-n}\alpha^{k}_{n}+\alpha^{k}_{-n}\psi^{k}_{n}}.$$
 
 Now, we use the same method of definition of $K, K'$ and $dK'$ in Section \ref{2.2.4}, we build the linear maps considering local linear maps defined on a covering $\{U_{\alpha}\}_{\alpha\in I}$ of $M$, hence we define 
$$ L_{\alpha}|_{U_{\alpha}}=\sum_{n\in\mathbb{Z}_{+}}\sum_{0\leq k\leq 2l}\alpha^{k}_{-n}\alpha^{k}_{n},\qquad L_{\beta}=\sum_{n\in\mathbb{Z}_{+}}\sum_{0\leq k\leq 2l}\beta^{k}_{-n}\beta^{k}_{n} ,\qquad L_{\psi}|_{U_{\alpha}}=\sum_{n\in\mathbb{Z}_{+}}\sum_{0\leq k\leq 2l}n\psi^{k}_{-n}\psi^{k}_{n}.$$

\begin{definition}\label{4.2}   The momentun operator is defined by 
$$P:=L_{\beta}-L_{\alpha}-L_{\psi}-\frac{\dim M}{24}\in \End (\V_{R}).$$
  \end{definition}
  
 In the finite dimensional case the extra term is given by $c_{1}=\sum_{r\in\mathcal{A}}\sum_{k=1}^{d_{r}}r$ (see \eqref{c1}), this term does not make sense in this case. Instead, we consider the factor ${\dim M}/{24}$, this factor fix the modular invariance of the Witten genus. 
 
 \begin{remark} In the physics literature $\dim M/24$ appears as a zeta regularization of $\sum_{n\in\ZZ}\sum_{k=1}^{2l}n$. In vertex algebras, ${\dim M}/{24}$ appears as the central charge of the Neveu-Schwarz twisted module or Ramond fermions.
 \end{remark}

\subsection{The index and the supersymmetry}\label{sec4.4}
 
We define an Hermitian product on the vector space of global sections $\Gamma(\V_{R})$ as follows
$$\langle v, w\rangle :=\int_{M}(v(x),w(x))dM,\qquad  v,w\in \Gamma(\V_{R}).$$
Where $dM$ is the volume form coming form the Riemannian metric on M. From Lemma \ref{ad}, the operator defined by $\slashed{D}_{N}:=Q-\slashed{D}_{M}$ is self adjoint, moreover $\slashed{D}_{M}$ is a twisted Dirac operator and its self adjoinest follows from the divergence for vector fields. Then, we obtain the following proposition

\begin{proposition}\label{i2} The operator $Q$ is self adjoint i.e.
$$\langle Q_{R}v,w\rangle =\langle v,Q_{R}w\rangle ,\qquad v,w\in \Gamma(\V).$$
\end{proposition}
 
 
We are motivated by  \eqref{equiin} to define \emph{the analytic equivariant index} of $Q_{R}$ as follows. 
$$\Ind (Q_{R},q):=\Tr (q^{P}|_{\ker (Q|_{\V^{+}_{R}})})-\Tr (q^{P}|_{\ker (Q|_{\V^{-}_{R}})}).$$
 where $q=e^{i\theta}\in\mathbb{C}^{\times}$. 
 
The proof of the following lemmas and theorem can be translated \emph{mutatis mutandis} from Lemmas \ref{operalg} and \ref{susy} and Theorem \ref{index}.

\begin{lemma} The operators satisfy the following relations

\emph{(i)}\; $[Q_{R} ,P]=0$

\emph{(ii)}\;  $Q^{2}_{R}=Q^{2}_{M}+L_{\alpha}+L_{\psi}$ 

\end{lemma}

\begin{lemma}[susy]
{$\ker (Q_{R})\subset V_{R}':=\bigtriangleup (M)\otimes S(N)\subset V$.}
\end{lemma}

And, we have the theorem that express the index in terms of a topological expression
    
    \begin{theorem}\label{index2}
    $$\Ind (Q_{R},q)=q^{-\frac{\dim M}{24}}\int_{M} A(M)\ch (\bigotimes_{n=1}^{\infty}S_{q^{n}}(TM_{\mathbb{C}}))=\frac{\Phi(M,q)}{\eta(q)^{dimM}}.$$
\end{theorem}

\end{document}